\documentclass[12pt, onecolumn]{IEEEtran}
\usepackage{cite}
\usepackage{pslatex,bm}
\usepackage{epsfig}
\usepackage{amsmath}
\usepackage{amssymb}
\usepackage{array}
\usepackage{multicol}
\usepackage{color}
\usepackage{url}
\usepackage{dsfont}
\usepackage{float}
\psfull

\newtheorem{Theo}{Theorem}
\newtheorem{lemma}{Lemma}
\newtheorem{rmk}{\sf Remark}
\newtheorem{definition}{Definition}
\newtheorem{coro}{Corollary}

\DeclareMathAlphabet{\mathpzc}{OT1}{pzc}{m}{it}

\renewcommand{\baselinestretch}{1.9}
\def\diag{{\rm diag}}
\def\tr{{\rm tr}}

\def\b0{{\mathbf 0}}

\begin{document}

% paper title
\title{Practical Design for Multiple-Antenna Cognitive Radio Networks with Coexistence Constraint}

%\author{Pin-Hsun Lin \hspace{1.8cm} Gabriel P. Villardi \hspace{1.8cm} Zhou Lan\hspace{1.8cm}Hiroshi Harada \\

\author{
\authorblockN{Pin-Hsun Lin \IEEEmembership{Member, IEEE}, Gabriel P. Villardi, \IEEEmembership{Senior Member, IEEE}, \\
Zhou Lan, \IEEEmembership{Member, IEEE}, Hiroshi Harada, \IEEEmembership{Member, IEEE} }%%%%%%%%%%%%%%%%%%%%%%%%%%%%%%%%%%%%%%%%%%%%%%%%%%%%%%%%%%%%%%%%%%%%%%%%%%%%%%
%Footnotes (Thanks finical support)
%%%%%%%%%%%%%%%%%%%%%%%%%%%%%%%%%%%%%%%%%%%%%%%%%%%%%%%%%%%%%%%%%%%%%%%%%%%%%%

\thanks{
Pin-Hsun Lin, Gabriel P. Villardi, Zhou Lan, and Hiroshi Harada are with the Smart Wireless Laboratory, National
Institute of Information and Communications Technology, Yokosuka,
Kanagawa, Japan. Emails: [pslin, gpvillardi, lan, harada]@nict.go.jp.} }

% make the title area
\maketitle \thispagestyle{empty} \vspace{-15mm}
{%\renewcommand{\baselinestretch}{0.8}
\begin{abstract}

In this paper we investigate the practical design for the multiple-antenna cognitive
radio (CR) networks sharing the geographically used or unused
spectrum. We consider a single cell network formed by the primary users (PU), which are half-duplex two-hop relay
channels and the secondary users (SU) are single user additive white
Gaussian noise channels. All transmitters and receivers are with
multiple antennas. In addition, the coexistence constraint which
requires PUs' coding schemes and rates unchanged with the
emergence of SU, should be satisfied. The contribution of this paper are twofold. First, we explicitly design the scheme to pair the SUs to the existing PUs in a single cell network. Second, we jointly design the nonlinear precoder, relay beamformer, and the transmitter and receiver beamformers to minimize the sum mean
square error of the SU system. In the first part, we derive an approximate relation between the relay ratio, chordal distance and strengths of the vector channels, and the transmit powers. Based on this relation, we are able to solve the optimal pairing between SUs and PUs efficiently where the metric is the sum rate of SUs' network. In the second part, considering the feasibility of implementation, we exploit the Tomlinson-Harashima precoding
instead of the dirty paper coding to mitigate the interference at
the SU receiver, which is the known side information at the SU
transmitter. To complete the design, we first approximate the optimization
problem as a convex one. Then we propose an iterative algorithm to
solve it with CVX. This joint design exploits all the
degrees of design freedom including the spatial diversity, power, and the side information of the
interference at transmitter, which may result in better SU's
performance than those without the joint design. To the best of our
knowledge, both the two parts have never been considered in the literature. Numerical results show that the proposed pairing scheme outperforms the greedy and random pairing with low complexity. Numerical results also show that even if all the
channel matrices are full rank, under which the simple zero forcing scheme
is infeasible, the proposed scheme can still work well.
\end{abstract}

\section{ Introduction}
Recently, many research efforts have been devoted to the studies on the
so-called cognitive radio (CR) technology \cite{Mitola_CR}, which
enables unlicensed (or secondary) users to dynamically sense and
locate unused spectrum segments and to communicate via these unused
spectrum segments. Most CR systems initially proposed in the
literature, e.g., \cite{Haykin_CR} and references therein, adopt the
interference avoidance (underlay) based approach, which requires
secondary users (SUs) to vacate the spectrum once the primary users (PUs) signals are sensed. The concept of
interference avoidance based CR has been standardized, e.g., the
IEEE 802.22 Wireless Regional Area Network (WRAN) standard \cite{80222} (802.22 or 802.22 WRAN
herein). It is aimed at using the CR techniques to allow the sharing of the
\textit{geographically unused} spectrum which is originally
allocated to the television multi-cast service, and do no harmful
interference to the incumbent operation, i.e., digital TV and analog
TV multi-casting. However, the interference avoidance based CR
demands fast and accurate sensing of spectrum holes, or it can not
operate in the region covered by PU \cite{80222}. Furthermore, it
may not yield the most efficient spectrum utilization since only one
user can access the specific frequency bands at any given time and/or geographical
location.
%To achieve a
%non-zero rate of simultaneous transmission for the SUs with PUs when
%their transmissions do not cause any rate degradation at PUs even
%when the latter employ only single-user decoders, Jovicic and
%Viswanath proposed in \cite{Jovicic_CR} a scheme that utilizes
%relaying by the secondary transmitter to overcome the interference
%caused by the simultaneous transmission of SU's message. For the SU,
%the interference caused by the received signals corresponding to
%PU's message at the secondary receiver (which is received from both
%the direct and the relay links) is mitigated by employing dirty
%paper coding (DPC) \cite{Jovicic_CR}\cite{Costa_DPC} at the
%secondary transmitter. This transmission scheme has been shown to be
%capacity-achieving when the channel gain between secondary
%transmitter and primary receiver is smaller than that between the
%secondary transmitter and receiver \cite{Jovicic_CR}.

In this paper we consider the practical design of the \textit{interference mitigating} (overlay) \cite{devroye_CR} multiple-antenna CR
systems \cite{devroye_CR}, where the primary system includes several two-hop relay channels in a single cell. Contrary to the CR systems in \cite{80222}
which can only use the geographically unused spectrum, here the
considered CR systems can additionally access the frequency bands which are
geographically used with the coexistence constraint. The detail will be explained in the next section. And the practical scenarios of the considered model may include the homogeneous and heterogeneous networks. In both kinds of networks, SUs help to relay PU's signals to use PU's spectrum as a tradeoff. Besides, an apparent benefit of SU's relay is that when the channel between PU transmitter and receiver is poor or unstable, PU's transmission can still be successful. On the other hand, several capacity results of the interference mitigation CR utilizing the dirty paper coding (DPC) are
discussed in \cite{Jovicic_CR}\cite{Wu_DPC_CR}\cite{Rini_DPC_CR}.
Although DPC is capacity achieving, the complexity makes it
infeasible for practical systems. Instead of using DPC, in this
paper we design the practical interference mitigation based CR by exploiting the
Tomlinson-Harashima precoding (THP) \cite{THP1}\cite{THP2}. Note
that DPC capacity can be achieved by the nested lattice coding
\cite{ShamaiMultibinning}. And that means the \textit{binning} operation
in the random coding scheme of DPC can be accomplished by the vector
quantization. On the other hand, THP uses a scalar quantization
and thus can be treated as a special case off DPC. Due to its
simplicity, THP is widely used in channels where the side information is known at the transmitter \cite{caire_channel_with_SI,ShamaiMultibinning,Shenouda_THP,Stankovic_THP},
albeit the modulo loss, the shaping loss, and the power loss
\cite{Yu_trellis_precoding} relative to DPC. On the other hand, spatial diversity is also exploited to avoid the interference
between CR and PR's transmission by designing linear precoders
and/or receiver beamformers
\cite{Zhang_ZF_CR}\!\!\cite{Luca_ZF_CR}\!\!\cite{Gharavol_BF_CR}. Also,
the relay assisted transmission is considered in CR
design \cite{Hamdi_CR}\!\!\cite{Zheng_CR}.

 Different from the above, we consider a much more complete and complicated model: several SU transmitter and receiver pairs coexist with the PU network where all PUs are assumed to be two-hop relay channels. All nodes are assumed to have multiple antennas. Before the SUs can simultaneously transmit with PUs in the same time-frequency slot, each SU should be paired to only one PU, vice versa, to form an interference mitigation cognitive channel. After that, the transmitters and receivers in each single cognitive channel should be designed. Upon solving these two problems, the main contributions of this paper are listed in the following.\\
1. We derive an approximate representation of the relay ratio in terms of the transmit powers of PU and SU networks and the chordal distance and the strengths of the vector channels. Based on this representation, we are able to design an explicit scheme to pair the SUs to the existing PUs in a single cell network with low complexity, while the sum rate of SUs' network is maximized. To the best of our knowledge, such pairing design for the overlay CR with coexistence constraint has never been considered in the literature. Numerical results show that the proposed pairing outperforms the greedy algorithm and random pairing with low complexity.\\
2. We design the SU transmitter and receiver in each cognitive channel by exploiting all the degrees of design freedom such as the side information at transmitter, the spatial diversity, and the transmit power to jointly design the non-linear precoder, transmit and receive beamformer, and the relay matrix, for the interference mitigation based CR with coexistence constraint. To the best of our knowledge, such joint design has never been considered in the literature. To attain this goal, we first form an optimization problem, where the objective is to minimize the sum mean square error (SMSE) of the CR system, subject to the coexistence and power constraints. However, since the problem is non-convex, we approximate it and propose an iterative algorithm to solve it by the CVX \cite{cvx}. Numerical results show that the performance of the proposed joint design outperforms that with the precoder designed by the generalized zero forcing, not to mention the traditional zero forcing. Especially when the channels are full rank, there is no null space can be used whence the traditional zero forcing scheme can not be applied.

The rest of the paper is organized as follows. In Section
\ref{Sec_system_model} we introduce the considered system model. In Section \ref{Sec_pairing} we solve the pairing problem between the SU and the PU's networks. The practical design problem of the SU transceiver is discussed in Section
\ref{Sec_practical_design}. In Section \ref{Sec_numerical_result} we illustrate the numerical results. Finally, Section \ref{Sec_conclusion} concludes
this paper.

\section{System model}\label{Sec_system_model}
The considered cognitive channel is shown in Fig. \ref{Fig_sys}. In this model, all the transmitters and receivers
are with multiple antennas. Assume that the primary system is a
half-duplex two-hop relay channel, i.e., in the odd time slots the base station transmits to the PU-TX, and in the even time slots the PU-TX relays to the
PU-RX. This kind of network plays an important role in reducing the
transmit power and/or extending the coverage of communication
networks, which is considered in several wireless standards, e.g.,
\cite{80216}, etc. Contrary to the CR systems in \cite{80222} which
can only use the geographically unused frequency bands, here the considered
CR systems can access the frequency bands both geographically used
or unused with the \textit{coexistence constraint}, which
is defined as

\begin{definition}\textit{Coexistence constraint}\cite{Jovicic_CR}: Under CR's own transmission,
there exists a sequence of $(2^{nR_1} , n)$-codes that can be decoded with a single-user decoder
at the primary receiver with vanishing error probability as $n\rightarrow\infty$.
\end{definition}
Thus such coexisted secondary systems can increase the spectral
efficiency in a predefined time/frequency/
geo-location slot of many
of the state-of-the-art wireless communication
standards/technologies and the design of such secondary systems is
worth investigating. The way to attain the coexistence is explained
in the following. SU-TX can overhear the signal transmitted by the
base station due to the wireless characteristic. As long as the
channel between the base station to the SU-TX is no worse than that
between the base station and PU-TX, this signal can be decoded
successfully at the SU-TX prior to the transmission of PR-TX. Thus
this signal can be treated as a non-causally known side information
at SU-TX. The importance of the non-causality are twofold. First,
SU-TX can use this non-causally known signal to do the decode and
forward relay to maintain the \!SINR at PU-RX, while SU-TX is simultaneous transmitting his own signal by the same spectrum geographically used by PU.
Second, with the knowledge of the non-causal side information, SU-TX
can use the advanced coding schemes such as DPC
\cite{Jovicic_CR}\!\!\cite{Costa_DPC}\!\cite{pslin_CR} when there is full
channel state information at the \!transmitter (CSIT) or
Gelfand-Pinsker coding \!\!\cite{GPC}\cite{PSLIN_CRRX2} for the partial CSIT
\cite{pslin_CR} case. This kind of CR is coined
by the \!\textit{interference mitigation CR} \!\!\cite{devroye_CR}, with the most commonly used \textit{interference avoidance CR},
are rigorously defined in the following
\begin{definition}\textit{Interference mitigation based CR}\cite{devroye_CR}:
With the non-causal knowledge of PU, the CR can simultaneously
operate in the spectrum occupied by PU without violating the
coexistence constraint.
\end{definition}
\begin{definition}\textit{Interference avoidance based CR}\cite{devroye_CR}: CR systems are active only when the PU is
silent.\\
\end{definition}
\begin{figure}
\centering \epsfig{file=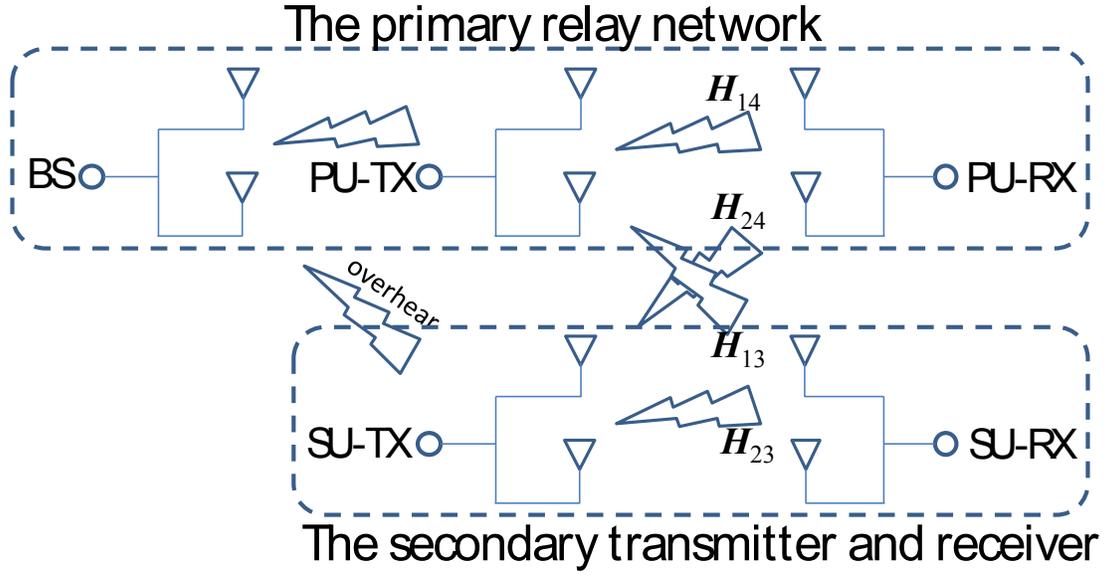 , width=.8\textwidth,
angle=0} \caption{The considered cognitive channel.} \label{Fig_sys}
\end{figure}
Although DPC is capacity achieving for the interference mitigation based CR channel \cite{Jovicic_CR}, the complexity
\cite{SC_JSAC}\cite{Erez_DPC_code_design} prohibits its use in practical systems widespreadly. Considering the implemental feasibility,
in this paper we design the interference mitigation CR by exploiting
the Tomlinson-Harashima precoding (THP) \cite{THP1}\cite{THP2}
instead of DPC. In this paper we jointly design the nonlinear
precoder, the transmitter and receiver beamformers, and the relay
matrix to implement the interference mitigation CR. Note that such
joint design exploits all the resources of side information at
transmitter, spatial diversity, and power.

The received signals at SU and PU's receivers can be respectively
represented as\footnote{In this paper, lower and upper case bold
alphabets denote vectors and matrices, respectively. The $i$th element of vector $\bm a$ is
denoted by $a_i$. And the element at the $i$th row and $j$th column
of the matrix $\bm A$ is $a_{ij}$. The superscript $(.)^T$ and
$(.)^H$ denotes the transpose and transpose complex conjugate. The
superscript $A^{(k)}$ denotes the value of $A$ in the $k$th
iteration. vec$(\bm A)$ stacks the columns of $\bm A$ as a super
vector. $|\bm A|$ denotes the determinant of $\bm A$. $\bm I_n$
denotes the $n$-dim identity matrix. Re$\{x\}$ and Im$\{x\}$ are the real and imaginary parts of the complex signal $x$, respectively. $A\otimes B$ denotes the Kronecker product. $(x)^+=\max\{0,x\}$.}
\begin{align}
\bm y_S&={\bm H_{23}}\bm x_S+{\bm H_{13}}\bm {x}_{P}+\bm n_S, \label{EQ_main}\\
\bm y_P&={\bm H_{14}}\bm {x}_{P}+{\bm H_{24}}\bm x_S+\bm n_P, \label{EQ_main2}
\end{align}
where $\bm x_S\in\mathds{C}^{N_T}$ is SU's transmit signal, $\bm
{x}_P\in CN(\bm 0, \Sigma_{\bm x_P})$ is PU's transmit signal where $\Sigma_{\bm x_P}\in\mathds{C}^{N_T\times N_T}$; $\bm H_{23}$ and
$\bm H_{14}$ are the channels between the transmitters of the SU and
PU to their designated receiver, respectively; $\bm
H_{24}$ and $\bm H_{13}$ are the channels between SU's transmitter
to PU's receiver and PU's transmitter to SU's receiver,
respectively; $\bm n_S\sim CN(\bm 0, \bm I)$ and $\bm n_P\sim CN(\bm 0, \bm I)$ are the circularly symmetric
complex additive white Gaussian noises at
receivers of SU and PU's networks, respectively. All channel gains
in this paper are assumed to be static. We also assume that all
receivers perfectly know their channel gains respectively, which can
be easily attained by the training process. There are average power
constraints for both SU and PU's networks respectively
\begin{align}
\mbox{tr}(E\left[\bm x_{S}\bm x_{S}^H\right])\leq P_T, \mbox{ and tr}(E\left[\bm {x}_{P}\bm{x}_{P}^H\right])\leq P_P. \label{EQ_SU_power_constraint}
\end{align}
The detailed functions of the secondary transmitter and receiver in
Fig. \ref{Fig_sys} are shown in the upper part of Fig. \ref{Fig_sys2}.
\begin{figure}
\centering \epsfig{file=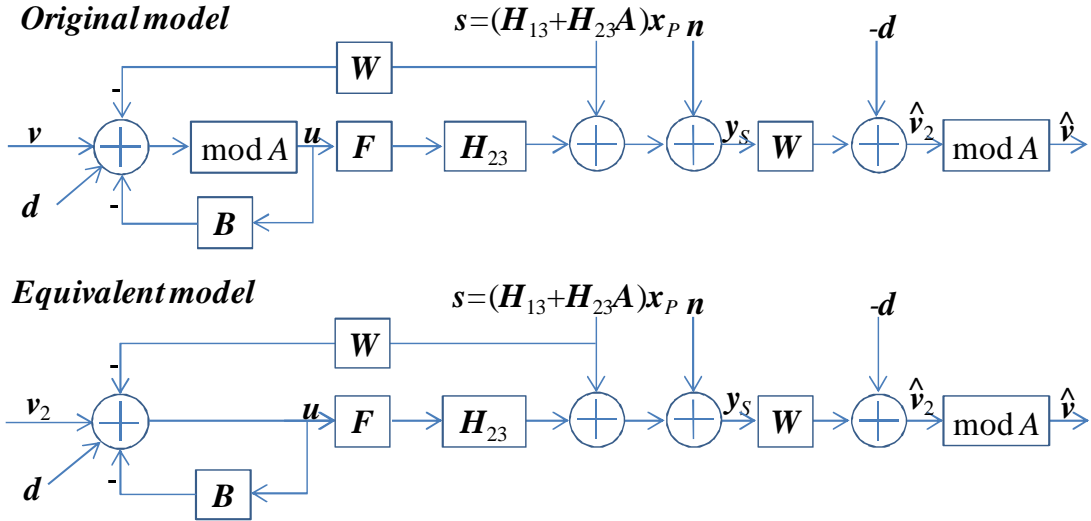 , width=.8\textwidth,
angle=0} \caption{The designed transmitter and receiver.}
\label{Fig_sys2}
\end{figure}
The channel input and output of the secondary systems can be respectively represented by
\begin{align}
\bm x_S &= \bm F\bm u+\bm A\bm x_P,\notag\\
\bm y_S &=\bm H_{23}\bm x_S+\bm H_{13}\bm x_P+\bm n_S=\bm
H_{23}\bm F\bm u+\bm s+\bm n_S,\label{EQ_SU_RX}
\end{align}
where $\bm F$ is the transmitter beamformer at SU-TX to transmit his own signal $\bm u$ and $\bm A$ is the relay matrix for relaying PU's signal $\bm x_P$.

\section{The pairing problem in a multiuser network}\label{Sec_pairing}
To accommodate the commercial systems, we assume that the PU-TXs are served by the base station with the orthogonal multiple access schemes such as the TDMA, FDMA, CDMA, SDMA, etc. We also assume that the transmissions among different PU-TX and PU-RX pairs use the same orthogonal multiple access scheme. And we assume the transmissions of the BS to PU-TXs and also the PU-TXs to PU-RXs are in different time slots. To consider the performance of SUs in a single cell such as the total throughput of the cell or the quality of services (QoS)\footnote{The reason to consider only the QoS of the secondary users is that the primary users's rates are not changed due to the coexistence constraint}, where several SUs coexist with PUs, to design a scheme properly pairing SUs to PUs is very critical to the system performance of SUs. For example, an improper pairing may cause the SU-TX to waste large portion of his power to relay PU's signal and results in a low SU's rate. In this section we first derive the approximate relay ratio in terms of the channel information and the transmit power, which is critical to the objective to be optimized. Then we formulate an integer programming problem which describes the considered pairing problem. After relaxing we propose an algorithm to numerically solve it. Note that in calculating the objective of the pairing problem, the proposed pairing scheme does not need to know all system parameters which should be designed as shown in Sec. \ref{Sec_practical_design}. Thus the design process is highly simplified.

\subsection{The metric of pairing}
In this section we will first give some observations from the SISO case. After that we will extend this relation to the MIMO case.
\subsubsection{The observation from the SISO case}
Recall that the coexistence constraint of the SISO case \cite{Jovicic_CR}\cite{pslin_CR} can be described by
%\begin{align}
%\frac{|h_{14}x_P+h_{24}\sqrt{\frac{\alpha P_T}{P_P}}x_P|^2}{P_N+|h_{24}|^2(1-\alpha)P_C}=\frac{|h_{14}|^2 P_P}{P_N}.
%\end{align}
%Then with the change of variables $P_1^{1/2}\triangleq|h_{14}|\sqrt{P_P}$, $P_2^{1/2}\triangleq|h_{24}|\sqrt{P_C}$, $P_3^{1/2}\triangleq|h_{23}|\sqrt{P_C}$, and some arrangements,
%we have
%\begin{align}
%(P_1P_2+P_2P_N)\alpha+2P_1^{1/2}P_2^{1/2}P_N\sqrt{\alpha}-P_1P_2=0,
%\end{align}
%and $\alpha$ can be solved as
%\begin{align}\label{EQ_alpha}
%\alpha=P_1P_2\left(\frac{-1+\sqrt{1+SNR_2+SNR_1SNR_2}}{(1+SNR_1)SNR_2}\right)^2.
%\end{align}
\begin{align}
\frac{\left|h_{14}x_P+h_{24}\sqrt{\frac{\alpha P_T}{P_P}}x_P\right|^2}{1+|h_{24}|^2(1-\alpha)P_T}=|h_{14}|^2 P_P,
\end{align}
where the left and right hand side are the signal to interference and noise ratio under SU is active and inactive, respectively. Note that the noise variances at both PU and SU RXs are set as unity without loss of generality. Also note that Gaussian signaling is used for PU and SU.
Then with the change of variables $P_1\triangleq|h_{14}|^2 P_P$, $P_2\triangleq|h_{24}|^2 P_T $, $P_3\triangleq|h_{23}|^2 P_T$, and some arrangements,
%we have
%\begin{align}
%(P_1P_2+P_2)\alpha+2P_1^{1/2}P_2^{1/2}\sqrt{\alpha}-P_1P_2=0,
%\end{align}
%and
$\alpha$ can be solved as
\begin{align}\label{EQ_alpha}
\alpha=\frac{P_1}{P_2}\left(\frac{-1+\sqrt{1+P_2+P_1P_2}}{1+P_1}\right)^2.
\end{align}

In the following we show how $h_{14}$, $h_{24}$, $P_P$, and $P_T$ affect SU's rate through $\alpha$. In
Fig. \ref{Fig_pairing_SISO_SU_rate_P1_PTfixed} we first show SU's rate versus $P_1$ with different
$P_2$'s and fixed $P_T$, i.e., the increasing of $P_2$ is due to the increasing of $|h_{24}|$.
It can be seen that SU's rate decreases with increasing $|h_{24}|$ and also with increasing $P_1$.
In the same figure we also fix $|h_{24}|$, i.e., the increasing of $P_2$
is due to the increasing of $P_T$. We can see that SU's rate increasing with increasing $P_T$ when $|h_{24}|$
is fixed. Thus we may summarize that except increasing $P_T$, increasing either $h_{14}$, $h_{24}$, or $P_P$
results in decreasing SU's rate.
\begin{figure}
\centering \epsfig{file=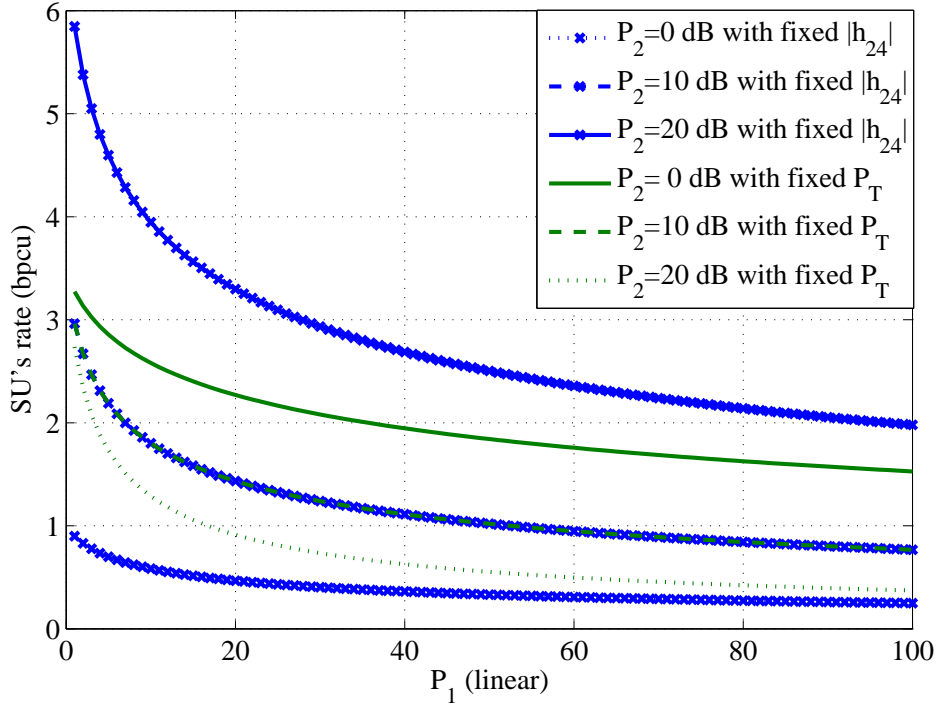 ,
width=.8\textwidth, angle=0} \caption{Cases of SU's rate versus $P_1$ when
$P_T$ is fixed and when $h_{24}$ is fixed.} \label{Fig_pairing_SISO_SU_rate_P1_PTfixed}
\end{figure}
\subsubsection{The extension to the MIMO case}
In the MIMO case, there is no simple trichotomy law between two channels like that in the SISO case. So before solving the pairing problem, first we need to extract the useful properties in vector case
from the matrices $\bm H_{14}$ and $\bm H_{24}$. An intuitive way is to consider the \textit{spatial correlation} of the two subspaces represented by $\bm H_{14}$ and $\bm H_{24}$, which are the only two channels affecting the relay ratio as \eqref{EQ_alpha}. For example, the principal angle \cite{Golub_principal_angle} and the chordal distance \cite{Conway_chordal_distance}.
\begin{definition}\label{Def_principal_angle}
(Principal Angle): For any two nonzero subspaces $V,\,W\subseteq \mathds{C}^n$, the principal angles between $V$ and $W$ are recursively defined to be the numbers $0\leq \theta_{i}\leq \pi/2$ such that
\begin{align}
\cos\theta_{i}&=\max_{\bm v\in V_{i},\,\bm w\in W_{i},\,||\bm v||_2=||\bm w||_2=1}\bm v^H\bm w,\,\,i=1,\,\cdots,\,\min\{dim(V),dim(W)\},
\end{align}
$\bm v_{i}$ and $\bm w_{i}$ are the vectors that construct the $i$th principal angle $\theta_{i}$, $||\bm v_{i}||_2=||\bm w_{i}||_2=1$, $V_{i}=V_{i-1}\cap\bm v_{i-1}^{\perp}$ and $W_{i}=W_{i-1}\cap\bm w_{i-1}^{\perp}$. Furthermore, $\theta_{min}=\theta_{1}\leq \cdots \leq  \theta_{max}$.
\end{definition}

\begin{definition}\label{Def_chordal_dist}
(Chordal distance): The chordal distance between $V$ and $W$ with dimensions $p$ and $q$, respectively, is defined in terms of the principal angles $\theta_{i}$ as
\begin{align}
d_c(U,V) =\sqrt{\Sigma_{i=1}^{\min\{p,q\}}\sin^2\theta_{i}}.
\end{align}
\end{definition}
%However, from the observation of the SISO case, we can see that only the angles between the subspaces are not sufficient, i.e., the magnitudes of them are also important. Thus we may modify the principal angle and chordal distance by multiplying the norms of the two channel matrices.
%A way to directly take the magnitude into account is by considering the volume of the projection of the two matrices
%\begin{align}
%\mbox{Vol}(\bm A)=\sqrt{\det(\bm A\bm A^H)},
%\end{align}
%where $\bm A= \bm H_{14}\bm H_{24}$. Note that $\mbox{Vol}(\bm A)$ is the volume of the paralellepiped defined by the column vectors of $\bm A$ and can be seen as an implicit measure of the orthogonality between the column vectors of $\bm A$.

In the following we investigate how the spatial correlation between the two channels $\bm H_{14}$ and $\bm H_{24}$ and also the transmit powers $P_P$ and $P_T$ affect the relay ratio. The proof is given in Sec. \ref{Proof_Theo_chordal}. Later we use this relation to formulate the pairing problem.

\begin{Theo}\label{Th_chordal}
To satisfy the coexistence constraint, the relay ratio and other system parameters approximately follow
\begin{align}\label{EQ_alpha_appr}
\alpha_{a}=\left(\frac{d_{24,min}^2\mbox{tr}(\bm D_{14}^2)-\frac{1}{P_T}}{d_{24,max}^2\left(\frac{1}{\lambda_{max}(\Sigma_{\bm x_P})} (d_c^2\mbox{tr}(\bm D_{14}^2)+N_T) +d_{24,min}^2\mbox{tr}(\bm D_{14}^2)\right)}\right)^+
\end{align}
where $d_c$ is the chordal distance between the eigen-spaces of $\bm H_{14}$ and $\bm H_{24}$.\\
\end{Theo}

\begin{rmk}
Note that for the tractability, we assume the signal distribution at SU-TX is Gaussian, which will be explained in more detailed in Remark \ref{Remark_worst_noise}.
\end{rmk}
\begin{rmk}\label{remark_alpha} From Theorem \ref{Th_chordal} we can find that the chordal distance, the eigenvalue spread of $\bm H_{24}$, the strength of the subspaces $\mbox{tr}(\bm D_{14}^2)$, and also the transmit power are used to characterize the relation with $\alpha$ for satisfying the coexistence constraint. The interplay between these parameters are discussed in the following. From \eqref{EQ_alpha_appr} it is clear that to make PU's rate fixed after SU is active, the relay ratio is approximately inversely proportional to the square of chordal distance. This result is intuitively reasonable, i.e., when the distance between the two subspaces is larger, SU can use less relaying power to keep PU's rate unchanged. As an extreme example, when $\bm H_{24}$ is orthogonal to $\bm H_{14}$, SU-TX does not need to relay since no SU's signal will be received by PU-RX under the assumption that PU-RX uses the left singular vectors of $\bm H_{14}$ as the beamformer. Note that when $d_{24,min}\rightarrow 0$, $\alpha\rightarrow 0$. This is reasonable since there may exist directions causing much less interferences to PU's RX due to the small eigenvalues, and SU's TX can transmit his own signal in such directions to reduce the relay ratio, i.e., the power for transmitting his own signal can be increased. In contrast, when $d_{24,max}^2$ is small, there may be no directions with relatively small enough eigenvalues for SU's own transmission to reduce the interference to PU's RX. And thus SU wastes more power on relaying, i.e., $\alpha_a$ increases. Note also that PU's transmit power has the same effect to $\alpha_a$ as that in the SISO case. It can easily be seen that $\alpha_a$ is proportional to $\lambda_{max}(\Sigma_{\bm X_P})$, which is consistent to the SISO case by the fact that $\lambda_{max}(\Sigma_{\bm X_P})$ may increase with increasing $P_P$.
\end{rmk}
%\textit{Remark:} Note that in \eqref{EQ_proportiona} we do not
%consider $\bm H_{24}^4$. The reason is that {\red sinfce
%$\lambda_{max}(\bm Q)$ is an implicit function of $\mbox{tr}(\bm
%H_{24}^4)$, we may not be able to explicitly describe their
%relation.} Note also that both PU and SU's transmit powers affect
%$\alpha$ in the MIMO case, which is consistent to the SISO case.
\subsection{The pairing problem formulation and solving}
Based on the relation in Theorem \ref{Th_chordal}, we can form the following optimization problem with only the knowledge of the channel matrices and the transmit power, but no need to know other system parameters such as the precoder, transmit and receive beamformer, etc., which can not be attained before the pairing. Note that we use $(i,j)$ to denote the \textit{PU-SU pair}, i.e., the $i$th SU transceiver is paired
to the $j$th PU transceiver, and $\alpha_{a,ij}$ is the approximate relay ratio of the pair $(i,j)$. To simplify the notation, we use $\bm H_{23,i}$, $\bm F_i$, and $\bm W_i$ to respectively denote the channel between the SU-TX and SU-RX, the transmit and receive beamformer, where the PU-SU pair is formed by the $i$th SU transceiver.

\begin{coro}
The maximization of the sum rate of the SU's network after pairing can be approximated by
\begin{align}
\underset{\{(i,j)\}}\max\mathop{\Sigma}_{i=1}^M\mathop{\Sigma}_{k=1}^{N_T}\log\left(1+\left(1-\alpha_{a,ij}\right)P_{T_{i}}\lambda_{23,ik}^2\right),\label{EQ_Pairing_opt1}
\end{align}
where there are $M$ SU and PU transceivers, respectively, and each node has $N_T$ antennas.
\end{coro}
\begin{proof}
We may represent the sum rate of the SU's network as
\begin{align}
\mathop{\Sigma}_{i=1}^M\log|\bm I + \bm W_i\bm H_{23,i} \bm F_{i}\bm F_{i}^H\bm H_{23,i}^H\bm W_i^H|\overset{(a)}=& \mathop{\Sigma}_{i=1}^M\log|\bm I + (1-\alpha_i)P_{T_i}\bm W_i\bm H_{23,i} \tilde{\bm F}_{i}\tilde{\bm F}_{i}^H\bm H_{23,i}^H\bm W_i^H|\notag\\
\overset{(b)} \simeq & \mathop{\Sigma}_{i=1}^M\log(|\bm I+(1-\alpha_{i})P_{T_{i}}\Lambda_{23,i}^2|)\notag\\
\overset{(c)}\simeq & \mathop{\Sigma}_{i=1}^M\mathop{\Sigma}_{k=1}^{N_T}\!\log\!\left(\!1\!+\!\left(\!1\!-\alpha_{a,ij}\!\right)\!P_{T_{i}}\!\lambda_{23,ik}^2\!\right),
\end{align}
where in (a) we normalize $\bm F_k$ such that $\mbox{tr}(\tilde{\bm F}_k\tilde{\bm F}_k^H)=1$; in (b) we assume $\bm W_{i}=\bm U_{23,i}$, $\tilde{\bm F}_i=\bm V_{23,i}/\sqrt{N_T} $, where $\bm U_{23,i}$ and $\bm V_{23,i}$ are the left and right singular vectors of $\bm H_{23,i}$, i.e., $\bm H_{23,i}=\bm U_{23,i}\Lambda_{23,i}\bm V_{23,i}^H$. This selection of $\bm W_i$ and $\tilde{\bm F}_i$ is due to the fact that before the pairing operation, we do not know the exact $\{\bm H_{13,i}\}$, $\{\bm H_{23,i}\}$, $\{\bm H_{14,i}\}$, and $\{\bm H_{24,i}\}$. Thus we can not attain $\bm F$ and $\bm W$ by solving the optimization problem $\mathbf{P1}$ in Sec. \ref{Sec_practical_design}. And here we use this simple but reasonable choice to make the following derivation tractable; in (c) Theorem \ref{Th_chordal} is used.
\end{proof}

Now we can reform \eqref{EQ_Pairing_opt1} as an explicit optimization problem
\begin{align}
\underset{\{t_{ij}\}}\max & \mathop{\Sigma}_{i=1}^M\mathop{\Sigma}_{j=1}^M t_{ij} \mathop{\Sigma}_{k=1}^{N_T}\log\left(1+\left(1-\alpha_{a,ij}\right)P_{T_{i}}\lambda_{23,ik}^2\right),\,\,\,
s.t.  \mathop{\Sigma}_{i=1}^M t_{ij}=1,\,\,\forall\,\, j; \mathop{\Sigma}_{j=1}^M t_{ij}=1,\,\,\forall\,\, i; \,\,t_{ij}\in\{0,1\},\forall\,\, i,\,j,\notag
\end{align}
where the indicator $t_{ij}$ is 1 if the $i$th SU is paired with the $j$th PU. And the first and second constraints state that each PU is only paired by one SU, and each SU is only paired by one PU, respectively. Note that this problem is an integer programming one, which is computationally prohibited especially when $M$ is large. Thus we relax the third constraint as $t_{ij}\geq 0,\,\forall\,i,\,j$. In the following we resort to the dual method to solve the relaxed problem. By dualizing the first constraint, we have the Lagrangian as
\begin{align}
L(\bm T, \bm \mu)&\!=\!\mathop{\Sigma}_{i=1}^M\!\mathop{\Sigma}_{j=1}^M \!t_{ij} \! \mathop{\Sigma}_{k=1}^{N_T}\log\!\left(\!1\!+\!\left(\!1\!-\alpha_{a,ij}\!\right)P_{T_{i}}\!\lambda_{23,ik}^2\right)\!+\!
\mathop{\Sigma}^{M}_{j=1}\!\mu_j\left(\!1\!-\!\mathop{\Sigma}^{M}_{i=1}\!t_{ij}\!\right)\triangleq\! \mathop{\Sigma}_{i=1}^M\mathop{\Sigma}_{j=1}^M t_{ij} (R_{ij}-\mu_j)\!+\!\mathop{\Sigma}_{j=1}^M\mu_j,\label{EQ_X}
\end{align}
where $R_{ij}\triangleq \mathop{\Sigma}_{k=1}^{N_T}\log\!\left(\!1\!+\!\left(\!1\!-\alpha_{a,ij}\!\right)P_{T_{i}}\!\lambda_{23,ik}^2\right)$ and $\bm \mu\triangleq[\mu_1,\,\mu_2,\,\cdots,\,\mu_M]\in\mathds{R}^M$ are the dual variables and the $i$th row $j$th column of the indicating matrix $\bm T$ is $t_{ij}$. And the dual objective function is
\begin{align}
g(\bm \mu)=\underset{\bm T}\max\,L(\bm T,\bm \mu),\,\,\,s.t.\, \mathop{\Sigma}_{j=1}^M t_{ij}=1,\,\,\forall\,\, i; \,t_{ij}\geq 0,\,\forall\,i,\,j,\label{EQ_dual_obj}
\end{align}
and the dual problem is $\underset{\bm \mu}\min\,g\!(\!\bm \mu\!).$ Since $R_{i\!j}$ is independent of $\bm T$, we can solve the optimal $\bm T$ for \!\eqref{EQ_dual_obj}\! as
\begin{align}
t_{ij}^{opt}=\left\{\begin{array}{ll}1,&\,j^*=\underset{j=1,\cdots, M}{\arg\max} (R_{ij}-\mu_j)\\
0,&\mbox{otherwise},\end{array}\right.\label{EQ_pairing_from_dual}
\end{align}
for $i=1,\cdots, M$. In the final step we use the subgradient method \cite{Boyd_subgradient} to solve $\bm \mu$
\begin{align}
\mu_j^{(l+1)}=\mu_j^{(l)}-\gamma^{(l)}\left(1-\mathop{\Sigma}_{i=1}^M t_{ij}\right),\,j=1,\cdots,M,\label{EQ_subgradient}
\end{align}
where the superscript $(l)$ denotes the number of iteration and
$\{\gamma^{(l)}\}$ is the sequence of step sizes which should be
designed properly. With the new $\bm \mu$ in each iteration, the
subcarrier pairing can be updated from \eqref{EQ_pairing_from_dual}.
Note that \eqref{EQ_subgradient} will converge to the dual optimum
variables \cite{Boyd_subgradient}. The algorithm in Table
\ref{table_alg_pairing} summarized the above steps of solving $\bm
T$ and $\bm \mu$, where $S_j$ is the set of indices of non-zero
entries in the $j$th column of $\bm T$ and $Z$ is the set of indices
of columns with all zero entries of $\bm T$. Note that from line 6
to 13 we try solve the problem that some columns may be allocated
with more than one 1 by moving the additional 1s to the all zero
columns with least reduction of the objective.
 {\renewcommand{\baselinestretch}{1}
\begin{table}[ht]
\caption{The Proposed Algorithm for the optimal pairing of the PU and SU networks} \normalsize
\centering % used for centering table
\vspace{-0.5cm}\begin{tabular}{l} % centered columns (4 columns)
\line(1,0){460}\\ %inserts double horizontal lines
1: Initialize $\bm \mu^{(0)}$ \notag\\
2. \textbf{Repeat} \\
3: \hspace{0.5cm} Compute $\left\{R_{ij}^{(l)}\triangleq \mathop{\Sigma}_{k=1}^{N_T}\log\!\left(\!1\!+\!\left(\!1\!-\alpha_{a,ij}^{(l)}\!\right)P_{T_{i}}\!\lambda_{23,ik}^2\right)\right\}$\\
4: \hspace{0.5cm} Compute $\left\{t_{ij}^{opt,(l)}=1,\,j^*=\underset{j=1\cdots M}{\arg\max}\,\,(R_{ij}-\mu_j^{(l)});\,\,
t_{ij}^{opt,(l)}=0,\mbox{otherwise}\right\}$, for $i=1,\cdots, M$\\
5: \hspace{0.5cm} Compute $\mu_j^{(l+1)}=\mu_j^{(l)}-\gamma^{(l)}\left(1-\mathop{\Sigma}_{i=1}^M t_{ij}\right)$, for $j=1,\cdots, M$\\
6: \hspace{0.5cm} For $j=1\sim M$\\
7: \hspace{1cm}  $i^*=\arg\max_{i\in S_j}(R_{ij}-\mu_j)$, $S_j=S_j\setminus i^*$\\
8: \hspace{1cm}   While $|S_j|\geq 1$\\
9: \hspace{1.7cm}  $i^*=\arg\max_{i\in S_j}(R_{ij}-\mu_j)$, $k^*=\arg\min_{k\in Z}|\mu_k-\mu_j|$,\\
10:\hspace{1.7cm}   swap the entries $(i^*,j)$ and $(i^*,k^*)$\\
11:\hspace{1.7cm}   $S_j=S_j\setminus i^*$, $Z=Z\setminus k^*$\\
12: \hspace{0.9cm} End\\
13: \hspace{0.4cm} End\\
14: \textbf{Until} $||\bm \mu^{(l+1)}-\bm \mu^{(l)}||/||\bm \mu^{(l+1)}||<\epsilon$\\
15: Return $\{t_{ij}^{opt}\}$\\
\line(1,0){460}\\%inserts single line
\end{tabular}
\label{table_alg_pairing} % is used to refer this table in the text
\end{table}
}
\begin{rmk}\label{remark_complexity}
The complexity of brute force pairing is $O(M!)$. In contrast, the proposed scheme highly reduces the complexity to $O(M^2)$ for
calculating $R_{ij}$ in each iteration. Assume $T$ iteration is required for the algorithm to converge, the total complexity is
$O(TM^2)$, which is much more tractable and feasible, especially when $M$ is large.
\end{rmk}

\section{Practical design of the secondary systems}\label{Sec_practical_design}
After the pairing process discussed in the previous section, each SU transceiver knows the PU he should
coexist with and also all the channels in Fig. \ref{Fig_sys}. In this section we discuss our proposed scheme for solving the
relay matrix, the non-linear precoder, the transmitter and receiver beamformer. Considering the feasibility of implementation, we adopt the
THP as the precoder to mitigate the interference at the SU receiver instead of the dirty paper coding. The THP output for the $k$-th antenna can be
represented by
\begin{align}\label{EQ_u_k}
u_k=\left(v_k+d_k-\sum_{j=1}^{N_T} w_{kj}s_j -\sum_{j=1}^{k-1}b_{kj}
u_j\right)\mbox{ mod } A,\,\,k=1\sim N_T,
\end{align}
where we assume $v_k,\,\,k=1\sim N_T$ are mutually independent and selected from an $M$-QAM
constellation; the $k$th element of $\bm s$ is denoted by $s_k$ and $\bm s=(\bm H_{13}+\bm
H_{23}\bm A)\bm {x}_P$ is the non-causally known side information at
SU-TX and $\bm A$ is the relay matrix; the random dither $d_k\in\mathds{C}$ is at the $k$-th antenna which is known by both the transmitter and the receiver before the transmission, where Re$\{d_k\}\sim\mbox{unif}(\frac{-A}{2},\frac{A}{2})$ and Im$\{d_k\}\sim\mbox{unif}(\frac{-A}{2},\frac{A}{2})$. Note that due to the insertion of the dither, the
elements of $\bm u$ are uncorrelated, i.e., $E[\bm u\bm u^H]=\bm I$, and are uniformly distributed over
the Voronoi region \cite{Fischer_precoding} which is proved as the Crypto Lemma \cite{ShamaiMultibinning}. The modulo operation is as following. For a vector $\bm g$, the $\mbox{mod } A$ operation $\bm g \mbox{ mod }A = \bm g'$ is applied element-wise to each
element $g_i$ of $\bm g$ such that $g_i'= g_i
-Q_A(g_i),\,\,\forall i$, where $Q_A(g_i)$
is the nearest multiple of $A$ to $g_i$. And the Voronoi region of $A$ is set as $[-A/2,\,A/2)$. Note that $\bm B$ is a strict lower triangular matrix, i.e.,
$b_{kj}=0,\,\forall k\leq j$, to perform the successive interference cancelation, inherited from its dual operation, i.e., the decision feedback equalizer at the receiver. Without loss of generality, we can
stack $u_1,\cdots, u_{N_T}$ in \eqref{EQ_u_k} in the vector form
\begin{align}
\bm u=(\bm v+\bm d-\bm W\bm s-\bm B \bm u)\mbox{ mod } A=(\bm I+\bm B)^{-1}(\bm v+\bm j+\bm d-\bm W\bm s)\triangleq \bm
C^{-1}(\bm v_2+\bm d-\bm W\bm s),\label{EQ_u}
\end{align}
where the second equality is from \cite{Fischer_precoding} and $\bm v_2\triangleq\bm v+\bm j$, $\bm j$ is the modulo indices. With $\bm j$, we may
represent the considered model as an equivalent one without the
modulo at transmitter, as shown in the lower part of Fig.
\ref{Fig_sys2}. In addition,
assume a moderate or high $M$ is used, then the
 \textit{power loss} due to the modulo operation can be neglected.
Then the estimated signal after the dither removal at receiver can
be represented by
\begin{align}
\hat{\bm v}_2 &=\bm W\bm y_S - \bm d= \bm{W}(\bm H_{23}\bm F\bm u+\bm s+\bm
n)- \bm d.\label{EQ_mod_output}
\end{align}
From \eqref{EQ_u} and \eqref{EQ_mod_output}, the estimated error can
be described by
\begin{align}
\bm e&\triangleq\hat{\bm v}_2-
\bm v_2=\bm W(\bm H_{23}\bm F\bm u+\bm s+\bm n)-\bm d
- \bm C\bm u-\bm W\bm s+\bm d=(\bm W\bm H_{23}\bm F- \bm C)\bm u+\bm W\bm n.
\end{align}
Then the sum mean square error (SMSE) among all antennas is
\begin{align}\label{EQ_Cov_mat_E}
\mbox{tr}(\bm E)&=\mbox{tr}(E[\bm e\bm e^H])=\mbox{tr}\left((\bm W\bm H_{23}\bm F-\bm C)(\bm W\bm H_{23}\bm
F-\bm C)^H+\bm W\bm W^H\right)=||\bm m||^2_2,
\end{align}
where in the last equality we define $\bm m\triangleq (\mbox{vec}(\bm W\bm H_{23}\bm F-\bm C)^T,\,\,\mbox{vec}(\bm W)^T)^T$.

Here we use the SMSE in \eqref{EQ_Cov_mat_E} as the performance
metric, instead of the sum rate of SUs as in Sec. \ref{Sec_pairing}.
In fact, maximizing SMSE is closely related to maximizing the sum
rate as following shows, while the later may be intractable. Note
that the error covariance matrices of the minimum MSE-decision
feedback equalization (MMSE-DFE) and MMSE-THP are the same \cite{Li_MSE_SINR}, under
the assumption that the former has correct previous decisions while
the latter has negligible precoding loss and valid linear model
\cite{Shenouda_THP}. In addition, we know that for an MMSE-DFE
system,
\[
\max\,\,\mathop{\Sigma}_{k=1}^{N_T}\log(1+\textsf{SINR}_k)=\max\,\,-\log\mathop{\Pi}_{k=1}^{N_T} \textsf{MSE}_k\geq -N_T\log\left(\min\,\,\mbox{tr}(\bm E)/N_T\right),
\]
where the equality is from
$\textsf{SINR}_k=1/\textsf{MSE}_k-1,\,k=1,\cdots,N_T$
\cite{Li_MSE_SINR} and the inequality comes from arithmetic and
geometric means and also \eqref{EQ_Cov_mat_E}. From the above,
minimizing the sum \textsf{MSE} maximizes the lower bound of the sum
rate of the SU's networks. Therefore we may claim that these two
metrics are consistent even in $\mathbf{P0}$ we have additional
variable $\bm A$ and constraints compared to the traditional MMSE-THP problem.

Then we can
form the optimization problem as
\begin{align}
\mathbf{P0:}\!\!\!\!\min_{\bm A,\,\bm B,\,\bm F,\,\bm W}&\, ||\bm m||^2_2\label{EQ_MSE_opt}\\
\mbox{s.t.}&\,\,\log\!\!\frac{|\bm I \!\!+ \!\!\bm H_{24}\bm F\!\!\bm F^H\bm H_{24}^H \!\!+\!\!(\bm H_{14}\!\!+\!\!\bm H_{24}\bm A)\!\Sigma_{\bm x_P}\!(\bm H_{14}\!\!+\!\!\bm H_{24}\bm A)^H|}{|\bm I+\bm H_{24}\bm F\bm F^H\bm H_{24}^H|}\geq \log(|\bm I+\bm H_{14}\Sigma_{\bm x_P}\bm H_{14}^H|),\label{EQ_coexistence_constraint3}\\
&\,\, \mbox{tr}(\bm F\bm F^H)+\mbox{tr}(\bm A \Sigma_{\bm x_P} \bm A^H)\leq P_T,\label{EQ_power_constraint_final}\\
&\,\,\bm B\mbox{ is a strictly lower triangular matrix,}\label{EQ_triangular_constraint}
\end{align}
where \eqref{EQ_coexistence_constraint3} is the coexistence
constraint to guarantee that PU's rate is not decreased when the SU
is active and the left hand side (LHS) is the rate of PU when the
SU is active; \eqref{EQ_power_constraint_final} is the
total power constraint at the SU transmitter, which is the sum power of SU's own signal tr$(\bm F\bm F^H)$ and that of the
relay signal tr$(\bm A\Sigma_{\bm x_P}\bm A^H)$. By eigenvalue
decomposition we have $\Sigma_{\bm x_P}=\bm V\Lambda_{\bm x_P}\bm
V^H$, where $\bm V$ is selected as the right singular vector of  $\bm H_{14}$, i.e., $\bm H_{14}=\bm U\bm S\bm V^H$ by singular value decomposition; $\Lambda_{\bm
x_P}=\diag\{P_1,\,P_2,\,\cdots,\,P_{N_T}\}$ where $P_i=\left(\mu-\frac{1}{\lambda_i^2}\right)^+$ is the water-filling
result according to $\bm H_{14}$, $\lambda_i$ is the
$i$th singular value of $\bm H_{14}$ and $\mu$ is chosen such that
$\Sigma_{i=1}^{N_T} P_i=P_p$.

\begin{rmk}\label{Remark_worst_noise} Note that the output of THP is uniformly distributed per two-dimension but not Gaussian. So the distribution of the received
signal at PU-RX is the convolution of uniform and Gaussian
distributions. And thus the PU's achievable rate is not easy to derive compared to the one with Gaussian input. For the tractability, we lower bound PU's rate
when SU is active by the concept that Gaussian noise is the worst
noise when the signal is Gaussian \cite{Cover_book}, i.e., treating SU's signal as Gaussian distributed results in the lowest PU's rate. Therefore solving $\bm A$, $\bm B$, $\bm
F$, and $\bm W$ with this lower bounded
constraint is more conservative and is thus feasible for the original problem.\\
\end{rmk}
By epigraph and Schur's Complement Theorem
\cite{Horn_matrix_analysis} \footnote{$\left(
                                         \begin{array}{cc}
                                           \bm A & \bm B \\
                                           \bm B^H & \bm C \\
                                         \end{array}
                                       \right)\succeq\bm 0\Leftrightarrow \bm C-\bm B^H\bm A^{-1}\bm B\succeq \bm 0
$, where $\bm A\succ \bm 0$ and $\bm C$ is Hermitian.}, we can
transform $\mathbf{P0}$ as
\begin{align}
\mathbf{P1:}\!\!\!\!\min_{\bm A,\,\bm B,\,\bm F,\,\bm W,\,t_0}&\,\, t_0,\,\,\,
\mbox{s.t. }\eqref{EQ_coexistence_constraint3},\,\,\eqref{EQ_power_constraint_final},\,\,\eqref{EQ_triangular_constraint},\,\left(%
\begin{array}{cc}
  t_0 & \bm m^H \\
  \bm m & \bm I \\
\end{array}%
\right)\succeq \bm 0.\label{EQ_t0_LMI}
\end{align}

To proceed, in the first step we simplify
\eqref{EQ_triangular_constraint} and solve $\bm W$ in the way
without loss of optimality, which can be attained due to the fact that $\bm W$ is at SU-RX and is independent of all the constraints. After some manipulations, we can have
the following lemma. The proof is provided
 in Section \ref{Sec_proof_W_B}.\\

\begin{lemma}\label{Lemma_P1}
Given $\bm F$, we can solve $\bm W$ and $\bm B$ as
\begin{align}
\bm W&=\bm C\bm F^H\bm H_{23}^H(\bm H_{23}\bm F\bm
F^H\bm H_{23}^H+ \bm I)^{-1},\label{EQ_W}\\
\bm B &= (\Delta_1+\Sigma_{i=1}^{N_T}\bm S_i^T\bm \mu_i^*\bm e_i^T)(\bm I+\bm \Delta_2),\label{EQ_B_solution}\\
\Delta_1&\triangleq \bm F^H\bm H_{23}^H\!\!(\bm H_{23}\bm F\bm F^H\bm H_{23}^H\!\!+\!\! \bm I)^{-1}\bm H_{23}\bm F\!\!-\!\!\bm I,\label{EQ_Delta1}\\
\Delta_2&\triangleq \bm F^H\bm H_{23}^H\bm H_{23}\bm
F,\label{EQ_Delta2}
\end{align}
where $\bm\mu_i$, $i=1,\cdots, N_T$ are the Lagrange multipliers of \eqref{EQ_triangular_constraint} which can be solved
from
\begin{align}
(1\!+\!\Delta_{2,kk})\bm S_k\bm S_k^T\bm \mu_k^*\! +\! \bm
S_k\underset{i\neq k}\Sigma\Delta_{2,ik}\bm S_i^T\bm \mu_i^*\!+\!\bm
c_k\!=\!\bm 0,\,\,k=1\sim N_T,\label{EQ_mu_equations}
\end{align}
$\Delta_{2,ij}$ is the entry at the $i$-th row and $j$-th column of $\Delta_2$,
$\bm S_k = [\bm 1_k,\,\bm 0_{k\times (N_T-k)}]$, $\bm c_k \triangleq
\bm S_k\Delta_1(\bm I+
\Delta_2)\bm e_k$, and $\bm e_k$ is the $k$-th column of the identity matrix.\\
\end{lemma}
Meanwhile, the coexistence constraint \eqref{EQ_coexistence_constraint3} can be approximated by the following lemma.\\

\begin{lemma} \label{Lemma_Taylor}
Given $ \tilde{\bm F}_2$, we can approximate
\eqref{EQ_coexistence_constraint3} by
\begin{align}&\log(1+t_1+t_2)-\mbox{tr}\left((\bm I+\bm H_{24}\tilde{\bm F}_2\bm H_{24}^H)^{-1}\bm H_{24}\bm F_2\bm H_{24}^H\right)\geq c_0,\label{EQ_Coexistence_constraint_taylor1}\\
&\left(
  \begin{array}{cc}
   t_1  & \mbox{vec}(\bm H_{24}\bm F)^H \\
   \mbox{vec}(\bm H_{24}\bm F)  & \bm I \\
  \end{array}
\right)\succeq\bm 0,\label{EQ_t1_LMI}\\
& \left[
  \begin{array}{cc}
    t_2 & \mbox{vec}(\bm H_{14}+\bm H_{24}\bm A)^H \\
    \mbox{vec}(\bm H_{14}+\bm H_{24}\bm A) & \bm I_2\otimes \Sigma_{\bm x_P} \\
  \end{array}
\right]\succeq\bm 0, \label{EQ_t2_LMI}\\
&\left[
  \begin{array}{cc}
    \bm I & \bm F \\
    \bm F^H & \bm F_2 \\
  \end{array}
\right]\succeq\bm 0, \label{EQ_A3_LMI}
\end{align}
where $c_0=\log|\bm I+\bm H_{14}\Sigma_{\bm x_P}\bm
H_{14}^H|+\log|\bm I+\bm H_{24} \tilde{\bm F}_2\bm
H_{24}^H|-\mbox{tr}\left((\bm I+\bm H_{24}\tilde{\bm F}_2 \bm
H_{24}^H)^{-1}\bm H_{24}\tilde{\bm F}_2\bm H_{24}^H\right)$.\\
\end{lemma}

The proof is provided in Section \ref{Sec_Proof_lemma_Taylor}.
To summarize the above, we can formulate a relaxed problem as
\begin{align}
\mathbf{P2:}\!\!\!\!\min_{\bm A,\,\bm F,\,\bm F_2,\,t_0,\,t_1,\,t_2}&\,\, t_0\mbox{, s.t.
}\eqref{EQ_t0_LMI},\,\,\eqref{EQ_W},\,\,\eqref{EQ_B_solution},\,\,\eqref{EQ_Coexistence_constraint_taylor1},\,\,\eqref{EQ_t1_LMI},\,\,\eqref{EQ_t2_LMI},\,\,\eqref{EQ_A3_LMI}.\notag
\end{align}
To proceed, we need to check the convexity of the constraints first. Note that $c_0$ is a constant. Thus
\eqref{EQ_Coexistence_constraint_taylor1} is convex of the slack
variables $t_1$, $t_2$, and $\bm F_2$ and thus can be solved by CVX
\cite{cvx}. In addition, since \eqref{EQ_t1_LMI}, \eqref{EQ_t2_LMI},
and \eqref{EQ_A3_LMI} are linear matrix inequalities, they also can
be solved by CVX. Note also that from \eqref{EQ_B_solution} $\bm B$
is not convex of $\bm F$. However, we may fix $\bm B$ and $\bm W$
first, and solve $\bm A$ and $\bm F$, which is then a convex problem. Thus
$\mathbf{P2}$ can be solved by CVX with fixed $\bm B$ and $\bm W$.
After that we substitute the solved $\bm F$ into
\eqref{EQ_B_solution} and \eqref{EQ_W} to find $\bm B$ and $\bm W$,
respectively. And we repeat this procedure iteratively. As a result
we propose the following iterative algorithm in Table
\ref{table_alg} to solve $\bm A,\,\bm B,\,\bm F$, and $\bm W$. Note
that from the numerical results we observe that the original
coexistence constraint \eqref{EQ_coexistence_constraint3} may be
violated during the iteration, which is due to the relaxation (however, the relaxed constraints
\eqref{EQ_Coexistence_constraint_taylor1}, \eqref{EQ_t1_LMI},
\eqref{EQ_t2_LMI}, and \eqref{EQ_A3_LMI} are all satisfied during the
iterations.) We also observe that the LHS of
\eqref{EQ_coexistence_constraint3} may decrease during
the iteration. Thus we insert an additional constraint as Step 6 into the algorithm to
help to terminate it.\\

\begin{table}[ht]
\caption{The Proposed Algorithm for solving the precoding matrix $\bm B$, the relaying matrix $\bm A$, the transmit beamformer $\bm F$, and the receive beamformer $\bm W$.} \normalsize
\centering % used for centering table
\vspace{-0.5cm}
\begin{tabular}{l} % centered columns (4 columns)
\line(1,0){480}\\ %inserts double horizontal lines
1: Initialize $\bm A^{(0)}=\bm 0,\,\,\bm F^{(0)}=\bm 0$ (thus $\bm B^{(0)}=\bm 0$),\,\,$\bm F_2^{(0)}=\bm 0$, $\tilde{\bm F}_2^{(0)}=\bm 0$, and $n=0$.\\
2: $\mathbf{Repeat}$\\
3: \hspace{0.5 cm}n=n+1\\
4: \hspace{0.5 cm}Given $\bm B^{(n-1)}$, $\bm W^{(n-1)}$, $\tilde{\bm F}_2^{(n-1)}$, solve $\bm A^{(n)}$ and $\bm F^{(n)}$ from $\mathbf{P2}$, \\
5: \hspace{0.5 cm}Update $\bm W^{(n)}=\bm C(\bm F^{(n)})^H\bm H_{23}^H(\bm H_{23}\bm F^{(n)}(\bm
F^{(n)})^H\bm H_{23}^H+ \bm I)^{-1}$ and\\
 \hspace{2.3 cm}$\bm B^{(n)}= (\Delta_1^{(n)}+\Sigma_{i=1}^{N_T}\bm S_i^T\bm \mu_i^{*,(n)}\bm e_i^T)(\bm I+\bm \Delta_2^{(n)})$, where $\Delta_1^{(n)}$, $\Delta_2^{(n)}$, and $\bm \mu_i^{*,(n)}$ are from \\
 \hspace{2.3 cm}\eqref{EQ_Delta1}, \eqref{EQ_Delta2}, \eqref{EQ_mu_equations}, respectively.\\
 \hspace{0.9 cm} Also update $ \tilde{\bm F}_2^{(n)}\leftarrow\bm F_2^{(n)}$.\\
6: $\mathbf{Until}$ The equality in \eqref{EQ_coexistence_constraint3} is valid.\\
\line(1,0){480}\\%inserts single line
\end{tabular}
\label{table_alg} % is used to refer this table in the text
\end{table}

\section{Numerical results}\label{Sec_numerical_result}
In this section, we first show the sum rate performance of the SU networks of the proposed
pairing scheme. Then we
use numerical results to illustrate the sum MSE performance of the proposed
joint design of the relay, THP, and transmitter/receiver
beamformers. We assume that the PU-TX, PU-RX, SU-TX, and SU-RX are
all equipped with 2 antennas. We also assume that with the knowledge
of $\bm H_{14}$, PU-TX can find $\Sigma_{\bm x_P}$ as explained in
Sec. \ref{Sec_practical_design}. The adopted channel model in the
pairing simulation is the general Kronecker product form channel
model \cite{Kermoal_MIMO_channel_model}: $\bm H=\sqrt{\rho}\bm
R_r^{1/2}\bm H_w \bm R_t^{1/2}$, where $\rho=1/(2 d^{\beta})$, $d$
is the distance between the uniformly generated PU and SU nodes and
$\beta$ is the path loss exponent which is set as 3; $\bm H_w\in
\mathds{C}^{2\times 2}$ is a zero mean unit variance i.i.d. complex
Gaussian matrix between the transmitter and receiver; $\bm R_t$ and
$\bm R_r$ are the transmit and receive correlation matrices which
are modeled by $[\bm R_t]_{ij}=\gamma_t^{|i-j|^2}$ and $[\bm
R_r]_{ij}=\gamma_r^{|i-j|^2}$, respectively, where $\gamma_t$ and
$\gamma_r$ are the correlation coefficients, respectively, which are
modeled as uniformly distributed variables within [0, 1]. The
multiplier $\bm \mu$ is randomly initialized within $[0,2]^M$ and
$\gamma^{(l)}$ is selected as $\gamma^{(l)}=0.05/\sqrt{l}$. In Fig.
\ref{Fig_pairing} we compare the sum rate performances of the
proposed method with greedy algorithm and random pairing. The greedy
algorithm randomly searches an SU and pairs it to PU in the greedy
sense. After removing the paired PU and SU nodes, this procedure is
done iteratively. The random pairing selects SU and pairs it to PU
both randomly. We consider three numbers of pairs in this example:
$M=$5, 10, and 20. We can find that the proposed scheme outperforms
the others in sum rates. Note that owing to the representation of
$\alpha$ in Theorem \ref{Th_chordal}, we can attain the pairing with
low complexity.

In the following we demonstrate the performance of the proposed transmission scheme. The SU-TX is assumed to know PU-TX's
selection of $\Sigma_{\bm x_P}$ to solve $\mathbf{P2}$. Without loss
of generality, we choose the vector channels as
\begin{align}\bm H_{13} &=\left(%
\begin{array}{cc}
  -0.7 + 0.28i & 1.82 + 0.2i \\
  -0.35 - 0.67i & -1.64 + 1.31i \\
\end{array}%
\right),\,\,
\bm H_{14} =\left(%
\begin{array}{cc}
  0.97 - 0.66i & -0.03 +
    0.77i \\
  -0.07 - 0.05i & 0.3 - 0.32i \\
\end{array}%
\right),\notag\\
\bm H_{23} &=\left(%
\begin{array}{cc}
  -1.12 + 0.57i & 0.41 +
    1.7i \\
  1.46 - 0.92i & 1.13 - 0.02i \\
\end{array}%
\right),\,\,
\bm H_{24} =\left(%
\begin{array}{cc}
  -0.77 + 0.38i & -0.09 - 0.41i \\
  -0.84 - 0.61i & 0.63 + 1.12i \\
\end{array}%
\right).\end{align}

Note that the above channels are all full rank. Thus it is
impossible to use the traditional zero forcing methods as
\cite{Zhang_ZF_CR}\cite{Luca_ZF_CR}\cite{Hamdi_ZF_CR}\cite{Heath_MIMO_CR} to avoid the
interferences from CR's signals to PU's receiver. Therefore the
coexistence constraint will never be valid. On the other hand, with the aid of
the considered relay as discussed previously, it is possible for the
CR systems to simultaneously transmit with PU. To overcome the caveat of the traditional ZF, we generalize it as \footnote{ \renewcommand{\baselinestretch}{2}
         For transmitting SU's own signal, it is better to avoid the direction to the PU RX. However, directly
         choosing the column vectors of $\bm H_{24}^{\perp}$ as the directions for transmitting SU's own signal may not be good when $\bm H_{23}$ has similar direction to $\bm H_{24}$.
          Similarly, to design $\bm A$, it is better to avoid the direction $\bm H_{23}$ such that not to interfere SU RX.
          Again, if $\bm H_{24}$ has similar direction to $\bm H_{23}$, this selection does not work well. So we resort
          to the linear combination of the desired subspace with its null space and then we optimize the coefficients
           to find the best directions. However, directions of $\bm F$ and $\bm A$ are designed almost
           separately, which may be inefficient compared to our method.}
    $\bm F=\sqrt{\alpha_{\bm F}P_T}\left(\sqrt{\gamma_{\bm F}}\mathbf{V}_1'\Sigma_1^{1/2}+\sqrt{1-\gamma_{\bm F}}\mathbf{V}_2'\Sigma_2^{1/2}\right),$ where $\bm V_1'$ and $\bm V_2'$ are formed by the singular vectors in $\bm V_1$ and $\bm V_2$ corresponding to non-zero singular values of $\Pi_{24}\bm H_{23}= \bm U_1\bm D_1\bm V_1^H$ and $\Pi_{24}^{\perp}\bm H_{23}= \bm U_2\bm D_2\bm V_2^H$, respectively, which are mutually orthogonal;
      $\Sigma_1$ and $\Sigma_2$ are derived from water-filling according to the singular values $\bm D_1$ and $\bm D_2$ respectively with
      $\mbox{tr}(\Sigma_1)=\mbox{tr}(\Sigma_2)=1$; $\alpha_{\bm F}\in[0,1]$  denotes the ratio of the total power allocated to the SU's signal transmission,
     and $\gamma_{\bm F}$ controls the power allocated in the directions of $\bm H_{24}$ and its null space; $\Pi_{24}=\bm H_{24}^H(\bm H_{24} \bm H_{24}^H)^{-1}\bm H_{24}$
      is a projection matrix projecting the operand to its column space and $\Pi_{24}^{\perp}=\bm I-\Pi_{24}$. Similarly, we can let $\bm A=\sqrt{\alpha_{\bm A}P_T}\left(\sqrt{\gamma_{\bm A}}\mathbf{V}_3'\Sigma_3^{1/2}+\sqrt{1-\gamma_{\bm A}}\mathbf{V}_4'\Sigma_4^{1/2}\right),$ where $\bm V_3'$ and $\bm V_4'$ are formed by the eigenvectors in $\bm V_3$ and $\bm V_4$ corresponding
          to non-zero singular values of $\Pi_{23}\bm H_{24}= \bm U_3\bm D_3\bm V_3^H$ and $\Pi_{23}^{\perp}\bm H_{24}= \bm U_4\bm D_4\bm V_4^H$, respectively.
    $\gamma_{\bm A}\in[0,1]$ controls the power allocated in the directions of $\bm H_{23}$ and its null space; $\Pi_{23}=\bm H_{23}^H(\bm H_{23} \bm H_{23}^H)^{-1}\bm H_{23}$
     is a projection matrix projecting the operand to its column space and $\Pi_{23}^{\perp}=\bm I-\Pi_{23}$. Note that $\Sigma_3$ and $\Sigma_4$ are derived from water-filling
     according to the singular values $\bm D_3$ and $\bm D_4$. And $\mbox{tr}(\Sigma_3)=\mbox{tr}(\Sigma_4)=1$. Then based on the line search, we can solve the three variables $\alpha_{\bm F},\,\gamma_{\bm F},\,\mbox{and }\gamma_{\bm A}$\footnote{\renewcommand{\baselinestretch}{2}
        When $\gamma_{\bm F}=0$, SU's transmission is zero forcing to PU's reception. Similarly, when $\gamma_{\bm A}=0$,
        the relay of PU's signal is zero forcing to SU's reception.}. In Fig.
\ref{Fig_SMSE_PT_with_W} we show the sum MSE performance versus
$P_T$ under different $P_P$s. From this figure we can find that the
minimum sum MSE may increase with increasing $P_P$. This is due to
the fact that when increasing $P_P$, i.e., SU-TX needs to allocate
more power for relaying PR's signal to make the coexistence
constraint valid. Thus less power can be used by SU-TX to transmit
its own signals and the sum MSE performance deteriorates with
increasing $P_P$. From the same figure we can easily see that the
proposed joint optimization outperforms the generalized ZF.  In
Fig. \ref{Fig_convergence} we show that how the sum MSE and the
coexistence constraint change with the number
 of iterations with $P_P=17$dB and $P_T=20$dB. Note that the coexistence constraint in the figure means the difference between the LHS and RHS of \eqref{EQ_coexistence_constraint3}. It can be seen that both the coexistence
 constraint and the sum MSE decreases with increasing number of iterations. These two subplots are consistent
 since as long as the coexistence constraint decreases, i.e., the LHS of \eqref{EQ_coexistence_constraint3}
 decreases, and more power is allocated to transmit SU's own signal. And thus SU can have better performance.
 As mentioned previously, the coexistence constraint decreases with the iterations, thus we need to check its
 crossover point to stop the iteration. In Fig.
\ref{Fig_CDF} we show the cumulative density function (CDF)
$F(\textsf{SMSE})$ of the SMSE under fast fading channels, where the
number of random channel realizations is $10^4$ and the four channels
are all set as i.i.d. Gaussian matrices with zero mean and unit
variance. From this figure we can find the most probable region of
the SMSE of the proposed scheme under random channels, which may not
be able to reflect in Fig. \ref{Fig_SMSE_PT_with_W}. Due to the
complexity of the ZF scheme, we do not show the CDF performance of it in
the same figure. In Fig. \ref{Fig_SER} we compare the symbol error rate (SER) performances between the proposed scheme and the generalized ZF. The SERs in both cases are averaged over two antennas. It can be easily seen that the performance loss owning to the non-optimized $\bm A,\,\bm B,\, \bm F,$ and $\bm W$ is large, which is about 5dB and 4.5 dB in the considered cases $P_T=50$ and $P_T=100$, respectively. On the other hand, the SER increases with $P_T$, which is consistent to that in Fig. \ref{Fig_SMSE_PT_with_W}. In Fig. \ref{Fig_RelayRatio_EffectivePower} we compare the relay ratios and the power SU can transmit his own signal versus $P_T$ between the proposed scheme and the generalized ZF, where the former is denoted by solid lines and the later is denoted by dashed lines. It can be seen that the power SU can use for transmitting his own signal increases with decreasing \!$P_P$ and increasing \!$P_T$, which is consistent to that of the SISO case in Fig.\! \ref{Fig_pairing_SISO_SU_rate_P1_PTfixed}
\begin{figure}[H]
\centering \epsfig{file=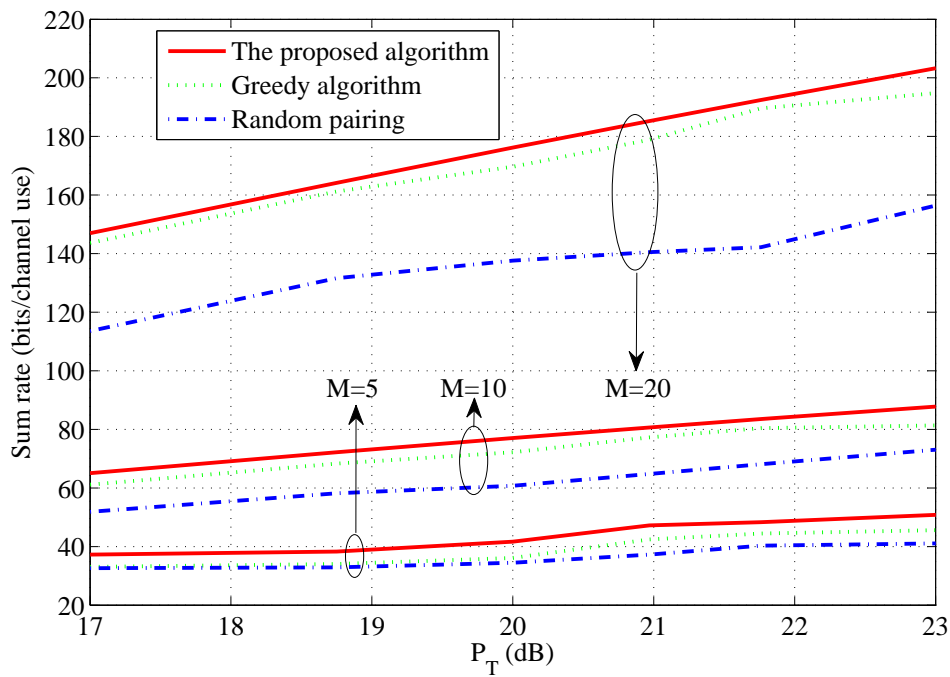 , width=.8\textwidth,
angle=0} \caption{Comparison of different pairing schemes.} \label{Fig_pairing}
\end{figure}
\begin{figure}[H]
\centering \epsfig{file=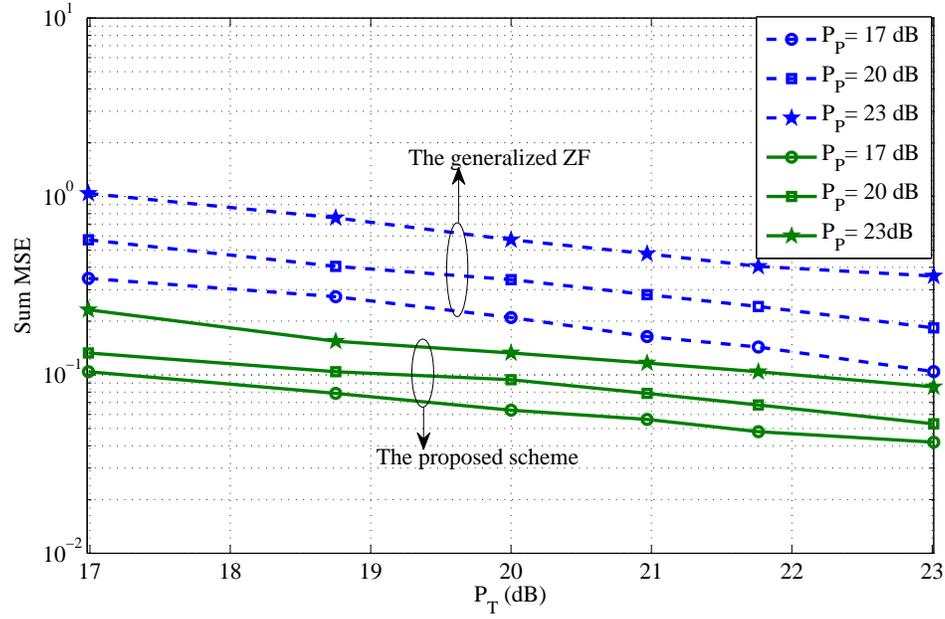 , width=.8\textwidth,
angle=0} \caption{The sum MSE performance versus $P_T$ under
 fixed channels with different $P_P$'s.} \label{Fig_SMSE_PT_with_W}
\end{figure}

\begin{figure}[H]
\centering \epsfig{file=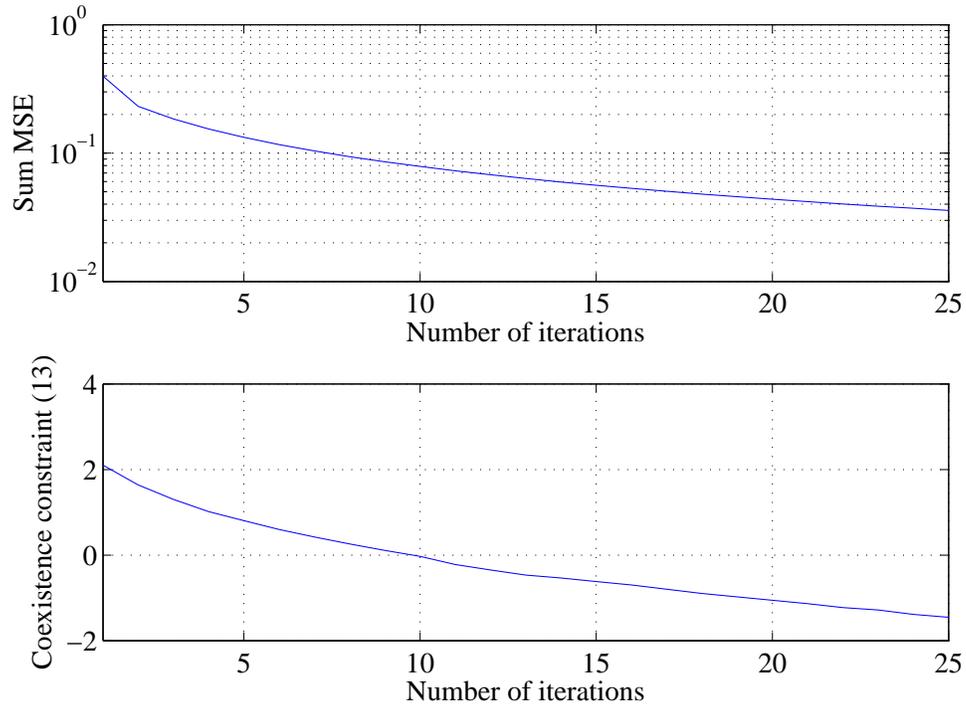 , width=.8\textwidth,
angle=0} \caption{The changes of the sum MSE and the coexistence constraint with the number of iterations.}
\label{Fig_convergence}
\end{figure}

\begin{figure}[H]
\centering \epsfig{file=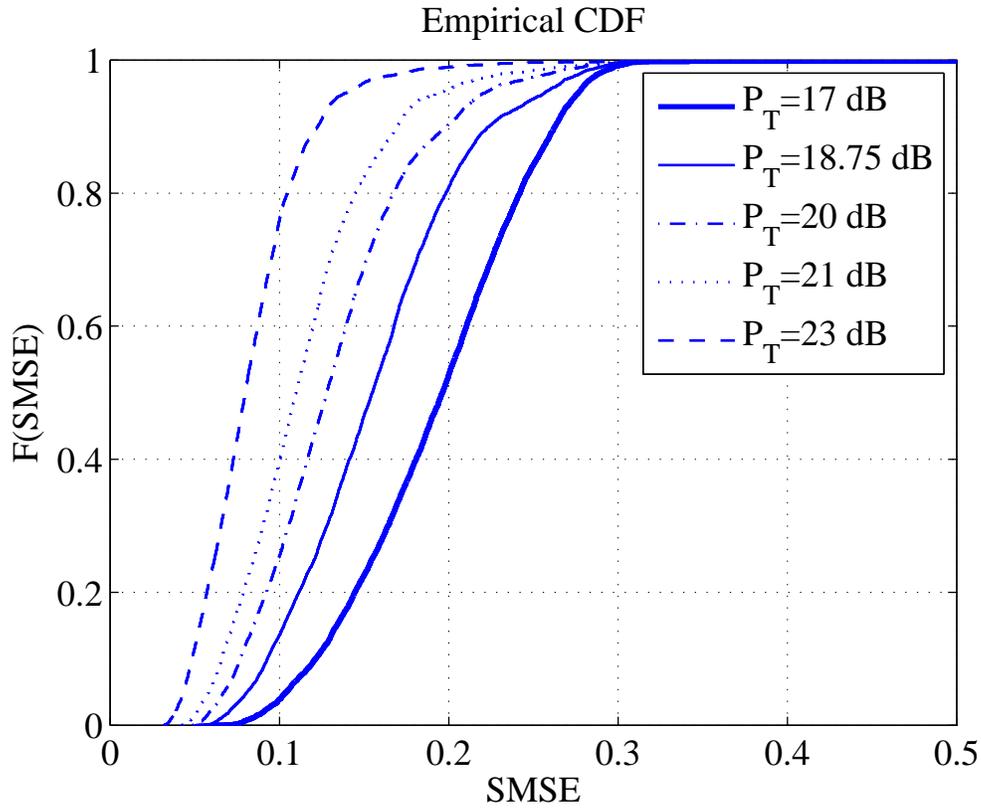 , width=.8\textwidth, angle=0}
\caption{The CDF of sum MSE performance under fast fading channels
with different $P_T$'s.} \label{Fig_CDF}
\end{figure}

\begin{figure}[H]
\centering \epsfig{file=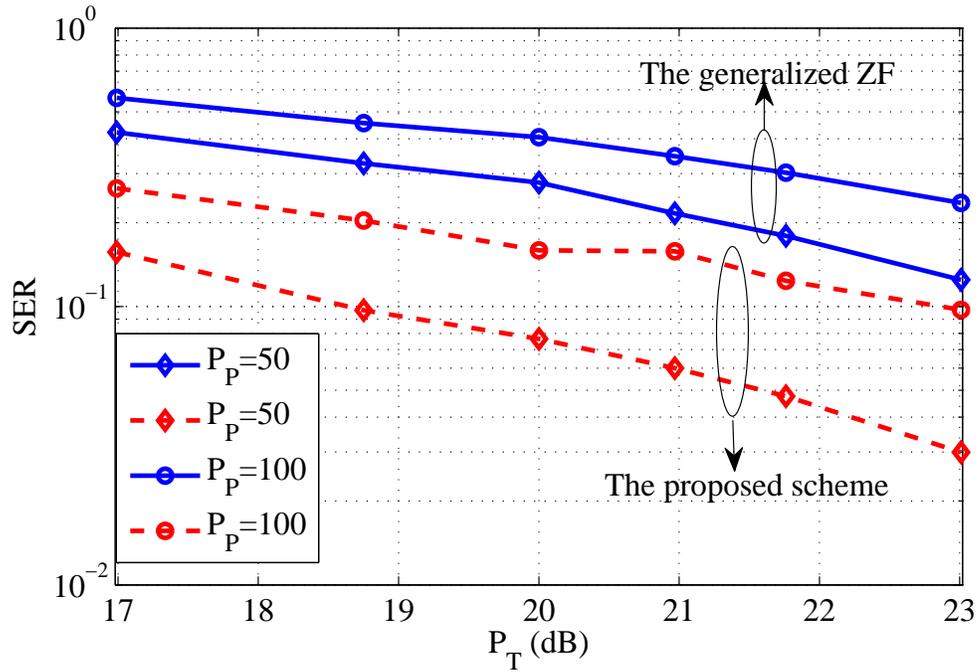 , width=.8\textwidth, angle=0}
\caption{The SER performance.} \label{Fig_SER}
\end{figure}

\begin{figure}[H]
\centering \epsfig{file=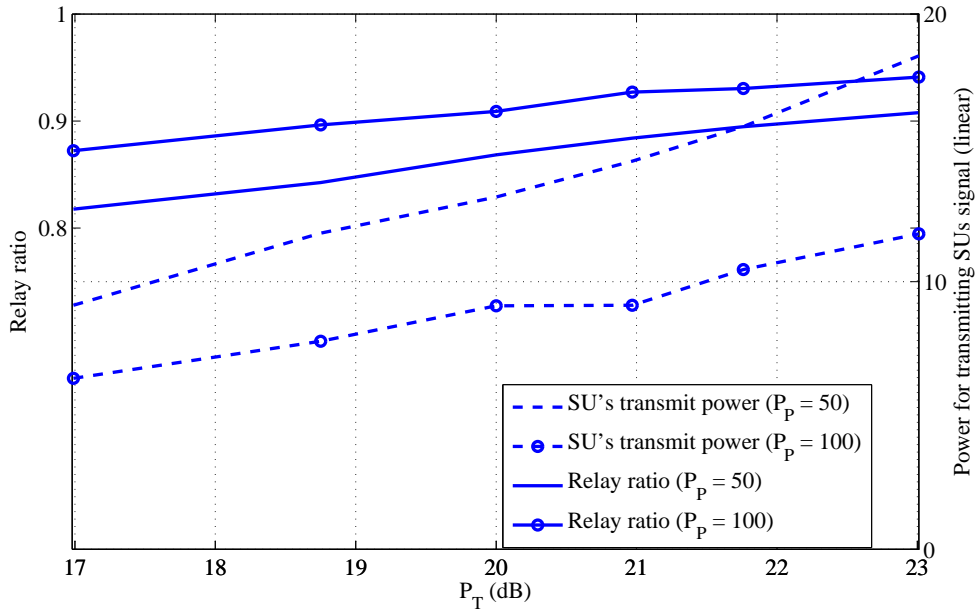 , width=.8\textwidth, angle=0}
\caption{The variation of relay ratio and SU's effective power for transmission with respect to SU's total power constraint.} \label{Fig_RelayRatio_EffectivePower}
\end{figure}

\section{Conclusion}\label{Sec_conclusion}
In this work we practically designed the multiple-antenna cognitive
radio (CR) networks with the coexistence constraint. The considered primary user (PU) is a half-duplex two-hop relay
channel and the secondary user (SU) is a single user additive white
Gaussian noise channel. The contribution of this paper are twofold. First,
we explicitly designed the scheme to pair the SUs to the existing PUs in a cellular network.
Second, we jointly designed the nonlinear precoder, relay beamformer, and the transmitter and
receiver beamformers to minimize the sum mean
square error of the SU system. In the first part, we solve the optimal pairing between SUs and PUs such
that the sum rate of SUs is maximized. To achieve it, we derive the relation between the relay ratio,
 chordal distance and strengths of the vector channels, and the transmit powers. In the second part,
 we adopt the Tomlinson-Harashima precoding instead of the dirty paper coding to mitigate the interference at
the SU receiver. To complete this design, we first approximated the optimization
problem as a convex one. Then we proposed an iterative algorithm to
solve it with CVX. This joint design exploits all the
degrees of design freedom including the spatial diversity, power, and the side information of the
interference at transmitter, which may result in better SU's
performance than those without the joint design. To the best of our
knowledge, both the two parts have never been considered in the literature.
Numerical results verifies the performance of the proposed schemes.

\section{Proof of Theorem \ref{Th_chordal}}\label{Proof_Theo_chordal}
\begin{proof}
In this proof we use the upper bound of the LHS of \eqref{EQ_coexistence_constraint3} to derive the approximate relay ratio $\alpha$. We rearrange the LHS of \eqref{EQ_coexistence_constraint3} as
\begin{align}
&\log|\bm I + (\bm I + \bm H_{24}\bm F\bm F^H\bm H_{24}^H)^{-1}(\bm H_{14}+\bm H_{24}\bm A)\Sigma_{\bm x_P}(\bm H_{14}+\bm H_{24}\bm A)^H|\notag\\
\overset{(a)} \leq & N_T\log\left(\frac{1}{N_T}\mbox{tr}\left(\bm I + (\bm I + \bm H_{24}\bm F\bm F^H\bm H_{24}^H)^{-1}(\bm H_{14}+\bm H_{24}\bm A)\Sigma_{\bm x_P}(\bm H_{14}+\bm H_{24}\bm A)^H\right)\right)\notag\\
\overset{(b)}\leq & N_T\log\left(1+\frac{1}{N_T}\lambda_{max}(\bm Q)\mbox{tr}\left((\bm H_{14}+\bm H_{24}\bm A)\Sigma_{\bm x_P}(\bm H_{14}+\bm H_{24}\bm A)^H\right)\right)\notag\\
\overset{(c)}\leq & N_T\log\left(1+\frac{1}{N_T}\lambda_{max}(\bm Q)\lambda_{max}(\Sigma_{\bm x_P})\left(\mbox{tr}(\bm A^H\bm H_{24}^H\bm H_{14})+\mbox{tr}(\bm H_{14}^H\bm H_{24}\bm A)+\mbox{tr}(\bm H_{14}\bm H_{14}^H)+\mbox{tr}(\bm H_{24}\bm A \bm A^H \bm H_{24}^H)\right)\right)\notag\\
\overset{(d)}\leq & N_T\log\left(1+\frac{1}{N_T}\lambda_{max}(\bm Q)\lambda_{max}(\Sigma_{\bm x_P})\left(|\mbox{tr}(\bm A^H\bm H_{24}^H\bm H_{14})|^2+1+\mbox{tr}(\bm H_{14}\bm H_{14}^H)+\mbox{tr}(\bm H_{24}\bm A \bm A^H \bm H_{24}^H)\right)\right)\notag\\
\overset{(e)}\leq & N_T\log\left(1+\frac{1}{N_T}\lambda_{max}(\bm Q)\lambda_{max}(\Sigma_{\bm x_P})\left({\mbox{tr}(\bm A\bm A^H)}{\mbox{tr}(\bm H_{24}^H\bm H_{14}\bm H_{14}^H\bm H_{24})}+c+\mbox{tr}(\bm H_{24}\bm A \bm A^H \bm H_{24}^H))\right)\right)\notag\\
%\overset{(e2)}\leq & N_T\log\left(1+\frac{1}{N_T}\lambda_{max}(\bm Q)\lambda_{max}(\Sigma_{\bm x_P})\left({\mbox{tr}(\bm A\bm A^H)}({\mbox{tr}(\bm H_{24}^H\bm H_{14}\bm H_{14}^H\bm H_{24})}+\mbox{tr}(\bm H_{24}H_{24}^H))+c\right)\right)\notag\\
\overset{(f)}\leq & N_T\log\!\left(\!1\!\!+\!\!\frac{1}{N_T}\lambda_{max}\!(\!\bm Q)\lambda_{max}(\Sigma_{\bm x_P})\left({\mbox{tr}(\bm A\bm A^H)}{\left(\mbox{tr}(\bm U_{14}^H\bm U_{24}\bm D_{24}^2\bm U_{24}^H\bm U_{14}\bm D_{14}^2)+\mbox{tr}(\bm D_{24}^2)\right)}\!+c\right)\!\!\right)\notag\\
\overset{(g)}\leq & N_T\log\!\left(\!1\!\!+\!\!\frac{1}{N_T}\lambda_{max}\!(\!\bm Q)\lambda_{max}(\Sigma_{\bm x_P})\left({\mbox{tr}(\bm A\bm A^H)}{\left(\mbox{tr}(\bm U_{14}^H\bm U_{24}\bm D_{24}^2\bm U_{24}^H\bm U_{14})\mbox{tr}(\bm D_{14}^2)+\mbox{tr}(\bm D_{24}^2)\right)}\!+c\right)\!\!\right)\notag\\
%\overset{(g)}\leq & \!N_T\!\!\log\!\!\Big(\!\!1\!\!+\!\!\frac{1}{N_T}\!\lambda_{max}(\bm Q)\!\lambda_{max}(\Sigma_{\bm x_P})\!\Big({\mbox{tr}(\bm A\bm A^H)}\!\left(\lambda_{max}(\!\bm U_{14}^H\!\bm U_{24}\!)\mbox{tr}(\!\bm D_{24}^2\!\bm U_{24}^H\!\bm U_{14}\!\bm D_{14}^2)+\mbox{tr}(\bm D_{24}^2)\right)\!\!+c\Big)\!\!\Big)\notag\\
\overset{(h)}\triangleq & N_T\log\!\left(\!1\!\!+\!\!\frac{1}{N_T}\lambda_{max}\!(\!\bm Q)\lambda_{max}(\Sigma_{\bm x_P})\left({\mbox{tr}(\bm A\bm A^H)}{\left(\mbox{tr}(\bm T\bm T^H\bm D_{24}^2)\mbox{tr}(\bm D_{14}^2)+\mbox{tr}(\bm D_{24}^2)\right)}\!+c\right)\!\!\right)\notag\\
\overset{(i)}\leq & N_T\log\!\left(\!1\!\!+\!\!\frac{1}{N_T}\lambda_{max}\!(\!\bm Q)\lambda_{max}(\Sigma_{\bm x_P})\left({\mbox{tr}(\bm A\bm A^H)}{\left(\Sigma_i d_{24,max}^2\lambda_i(\bm T\bm T^H))\mbox{tr}(\bm D_{14}^2)+\mbox{tr}(\bm D_{24}^2)\right)}\!+c\right)\!\!\right)\notag\\
\overset{(j)}= & N_T\log\!\left(\!1\!\!+\!\!\frac{1}{N_T}\lambda_{max}\!(\!\bm Q)\lambda_{max}(\Sigma_{\bm x_P})\left({\mbox{tr}(\bm A\bm A^H)}{\left(d_{24,max}^2 d_c^2\mbox{tr}(\bm D_{14}^2)+\mbox{tr}(\bm D_{24}^2)\right)}\!+c\right)\!\!\right)\notag\\
%\overset{(h)}\leq & N_T\log\!\Big(\!1\!+\!\frac{1}{N_T}\lambda_{max}(\bm Q)\lambda_{max}(\Sigma_{\bm x_P})\Big({\mbox{tr}(\bm A\bm A^H)}\left(\lambda_{max}(\bm U_{14}^H\bm U_{24})^2\mbox{tr}(\bm D_{14}^2\bm D_{24}^2)+\mbox{tr}(\bm D_{24}^2)\right)\!+c\Big)\!\!\Big)\notag\\
%\overset{(i)}\leq & N_T\log\Big(1+\frac{1}{N_T}\lambda_{max}(\bm Q)\lambda_{max}(\Sigma_{\bm x_P})\Big({\mbox{tr}(\bm A\bm A^H)}\left(\mbox{tr}(\bm D_{14}^2\bm D_{24}^2){\sin}^2{\theta_{max}}+\mbox{tr}(\bm D_{24}^2)\right)+c\Big)\!\!\Big)\notag\\
\overset{(k)}\simeq & N_T\log\Big(1+\frac{1}{N_T}\lambda_{max}(\bm Q)\Big({{\mbox{tr}(\bm A\Sigma_{\bm x_P}\bm A^H)}}\left(d_{24,max}^2 d_c^2\mbox{tr}(\bm D_{14}^2)+\mbox{tr}(\bm D_{24}^2)\right)+\lambda_{max}(\Sigma_{\bm x_P})c\Big)\!\!\Big)\notag\\
%\overset{(k)}= & N_T\log\Big(1+\frac{1}{N_T}\lambda_{max}(\bm Q)P_T\Big({\frac{\mbox{tr}(\bm A\Sigma_{\bm x_P}\bm A^H)}{P_T}}\left(\mbox{tr}(\bm D_{14}^2\bm D_{24}^2){\sin}^2{\theta_{max}}+\mbox{tr}(\bm D_{24}^2)\right)+\frac{\lambda_{max}(\Sigma_{\bm x_P})}{P_T}c\Big)\!\!\Big)\notag\\
\overset{(l)}\leq & N_T\log\Big(1+\frac{1}{N_T}\lambda_{max}(\bm Q)\Big(\alpha P_T\cdot d_{24,max}^2\left( d_c^2\mbox{tr}(\bm D_{14}^2)+N_T\right)+\lambda_{max}(\Sigma_{\bm x_P})c\Big)\!\!\Big)\notag\\
\overset{(m)}\leq & N_T\log\Big(1+\frac{1}{N_T}\frac{P_T}{1+\lambda_{min}(\bm H_{24}\bm F\bm F^H\bm H_{24}^H)}\Big(\alpha\cdot d_{24,max}^2\left( d_c^2\mbox{tr}(\bm D_{14}^2)+N_T\right)+\frac{\lambda_{max}(\Sigma_{\bm x_P})}{P_T}c\Big)\!\!\Big)\notag\\
\overset{(n)}\leq & N_T\log\Big(1+\frac{1}{N_T}\frac{P_T}{1+(1-\alpha)P_T\lambda_{min}(\bm H_{24}^H\bm H_{24})\lambda_{min}(\tilde{\bm F}\tilde{\bm F}^H)}\Big(\alpha\cdot d_{24,max}^2 d_c^2\left(\mbox{tr}(\bm D_{14}^2)+N_T\right)+\frac{\lambda_{max}(\Sigma_{\bm x_P})}{P_T}c\Big)\!\!\Big)\notag\\
\overset{(o)}= & N_T\log\Big(1+\frac{1}{N_T}\frac{\alpha P_T d_{24,max}^2}{1+(1-\alpha)P_T d_{24,min}^2\lambda_{min}(\tilde{\bm F}\tilde{\bm F}^H)}\Big(\left( d_c^2+c_1\right)\mbox{tr}(\bm D_{14}^2)+d_c^2N_T+c_1\Big)\!\!\Big),\label{EQ_approximated_PU2}
\end{align}
where (a) uses $\det(\bm M)\leq (\mbox{tr}(\bm M)/m)^m$, and $m$ is
the dimension of the square matrix $\bm M$; (b) is by the properties
$\mbox{tr}(\bm P\bm Q\bm P^H)\leq \lambda_{max}(\bm Q)\mbox{tr}(\bm
P\bm P^H)$ and $\mbox{tr}(\bm P\bm Q)=\mbox{tr}(\bm Q\bm P)$, where
$\bm Q\triangleq (\bm I + \bm H_{24}\bm F\bm F^H\bm H_{24}^H)^{-1}$
and $\bm P\triangleq (\bm H_{14}+\bm H_{24}\bm A)\Sigma_{\bm
x_P}^{1/2}$; (c) expands the product inside $\tr(.)$; (d) uses
$|\mbox{tr}(\bm A^H\bm H_{24}\bm H_{14})-1|^2\geq 0$; (e) uses
$|\mbox{tr}(\bm P\bm Q)|^2\leq \mbox{tr}(\bm P\bm P^H)\mbox{tr}(\bm
Q\bm Q^H)$, and we define $c=\mbox{tr}(\bm H_{14}\bm
H_{14}^H)+1$; in (f) we
decompose $\bm H_{24}=\bm U_{24}\bm D_{24}\bm V_{24}^H$ and $\bm
H_{14}=\bm U_{14}\bm D_{14}\bm V_{14}^H$ by singular value
decomposition, and use tr($\bm A\bm B$)=tr($\bm B\bm A$) and the property: if both $\bm P$ and $\bm Q$ are
positive semi-definite and Hermitian, then $0\leq \mbox{tr}(\bm P\bm
Q)\leq\mbox{tr}(\bm P)\mbox{tr}(\bm Q)$ \cite{Zhang_matrix_theory}; in (g) we use the property in (f) again; in (h)
we define $\bm T\triangleq\bm U_{24}^H\bm U_{14}$ with eigenvalues $\{\sin\theta_k\}$, where $\theta_k$ is the $k$th principal angle; in (i) we use $\mbox{tr}(\bm A\bm B)\leq\mathop{\Sigma}_{i=1}^n\lambda_i(\bm A)\lambda_i(\bm B)\leq \mathop{\Sigma}_{i=1}^n\lambda_{max}(\bm A)\lambda_i(\bm B)$; (j) is by definition of the chordal distance in Definition \ref{Def_chordal_dist}; in (k) we use
$\mbox{tr}(\bm P\bm Q\bm P^H)\leq \lambda_{max}(\bm Q)\mbox{tr}(\bm
P\bm P^H)$ reversely; in (l) $P_T$ is extracted out to form the term
$\frac{\mbox{tr}(\bm A\Sigma_{\bm x_P}\bm A^H)}{P_T}$, which is the
relay ratio by definition and by the fact $N_T d_{24,max}^2\geq\mbox{tr}(\bm D_{24}^2)$; (m) is by the fact that $\bm Q=\bm V(\bm I +\bm D)^{-1}\bm V^H$, where $\bm H_{24}\bm F\bm F^H\bm H_{24}^H=\bm V\bm D\bm V^H$ by eigenvalue decomposition; in (n) we set $\bm F=\sqrt{(1-\alpha)P_T}\tilde{\bm F}$ and use the two
properties \cite{Kumar_matrix_inequality}: 1. $\bm A\bm B$ and $\bm B\bm A$ have the same set of eigenvalues;
2. $\lambda_{min}(\bm A\bm B)\geq \lambda_{min}(\bm A)\lambda_{min}(\bm B)$ ; in (o) we define
$c_1\triangleq\frac{\lambda_{max}(\Sigma_{\bm x_P})}{\alpha P_T d_{24,max}^2}c$; the case of $d_{24,min}=0$
is discussed in \ref{remark_alpha}.

By equating the RHS of \eqref{EQ_approximated_PU2} to the upper bound of PU's rate when SU is not active, i.e.,
\begin{align}\label{EQ_PU_rate_UB}
\log|\bm I+\bm H_{14} \Sigma_{\bm x_P}\bm H_{14}^H|\leq N_T\log\left(1+\frac{1}{N_T}\lambda_{max}(\Sigma_{\bm x_P})\mbox{tr}(\bm D_{14}^2)\right),
\end{align}
we then can find \eqref{EQ_alpha_appr}. Note that this upper bound is by the properties in step (a) and (b). Note also that to derive \eqref{EQ_alpha_appr} we assume that $\tilde{\bm F}=\bm U_{23}$, where $\bm H_{23}=\bm U_{23}\bm D_{23}\bm U_{23}^H $ by eigenvalue decomposition. The assumption comes from the fact that before the pairing, optimal $\bm F$ is unknown and this selection makes the derivation tractable.

\section{Proof of Lemma \ref{Lemma_P1}}\label{Sec_proof_W_B}
We first rearrange the sum MSE in \eqref{EQ_Cov_mat_E} as
\begin{align}
\mbox{tr}(\bm E)=&\mbox{tr}\left(\bm W(\bm H_{23}\bm F\bm F^H\bm
H_{23}^H+ \bm I)\bm W^H
-\bm C\bm F^H\bm H_{23}^H\bm W^H-\bm W\bm H_{23}\bm F\bm C^H+\bm C\bm C^H\right).
\end{align}
Note that $\bm W$ is not included in any constraints. Thus the
Karush-Kuhn-Tucker (K.K.T.) condition which is the derivative of the
Lagrangian with respect to $\bm W$, is identical to
\begin{align}
\frac{\partial \mbox{tr}(\bm E)}{\partial \bm W^H}=\bm W(\bm
H_{23}\bm F\bm F^H\bm H_{23}^H+ \bm I)-\bm C\bm F^H\bm H_{23}^H=0
\end{align}
and thus $\bm W=\bm C\bm F^H\bm H_{23}^H(\bm H_{23}\bm F\bm F^H\bm
H_{23}^H+ \bm I)^{-1}$. After substituting $\bm W$ into
\eqref{EQ_Cov_mat_E}, we get
\begin{align}
&\mbox{tr}(\bm E)= \mbox{tr}\left(\bm C\bm C^H -\bm C\bm F^H\bm
H_{23}^H(\bm H_{23}\bm F\bm F^H\bm H_{23}^H+ \bm I)^{-1}\bm
H_{23}\bm F\bm C^H\right).\label{EQ_E_in_quadratic_C}
\end{align}
To solve $\bm C$ (or equivalently, $\bm B$), the lower triangular constraint should be taken into account, and which can be modeled by
\begin{align}\label{EQ_constraint_lower_triangular}
\bm S_k\bm B\bm  e_k = \bm 0_k,\,\,k=1\sim N_T.
\end{align}
Then we can form the Lagrangian of $\bm {P0}$ as
\begin{align}
\emph{L}=&\mbox{tr}\left((\bm I\!+\!\bm B)(\bm I\!+\!\bm B)^H\!-\!\!(\bm I\!+\!\bm B)\bm F^H\bm H_{23}^H\!(\bm H_{23}\bm F\bm F^H\bm H_{23}^H\!+\! \bm I)^{-1}\!\bm H_{23}\bm F\!(\bm I\!+\!\bm B)^H\right)-\Sigma_{k=1}^{N_T}\bm \mu_k^T\bm S_k\bm B\bm e_k\notag\\
&+\lambda\left\{\log|\bm I+\bm H_{14}\Sigma_{\bm x_P}\bm H_{14}^H|+\log|\bm I+\bm H_{24}\bm F\bm F^H\bm H_{24}^H|-\log\!\!|\bm I \!\!+ \!\!\bm H_{24}\bm F\!\!\bm F^H\bm H_{24}^H \!\!+\!\!(\bm H_{14}\!\!+\!\!\bm H_{24}\bm A)\!\Sigma_{\bm x_P}\!(\bm H_{14}\!\!+\!\!\bm H_{24}\bm A)^H|\right\}\notag\\
&+\gamma\left\{\mbox{tr}(\bm F\bm F^H)+\mbox{tr}(\bm A \Sigma_{\bm
x_P} \bm A^H)- P_T\right\},
\end{align}
where $\bm \mu_k\in\mathds{C}^k$, $\lambda\geq 0$, and $\gamma\geq
0$ are Lagrange multipliers. Note that the lower triangular
constraint is an equality constraint here, thus there is no
restriction on the sign of the multiplier $\bm \mu_i$. After equating the first
derivative of $\emph{L}$ with respect to $\bm B$ to zero, we have
\begin{align}
\frac{\partial\emph{L}}{\partial \bm B}=&(\bm I\!+\!\bm B)^H-\!\bm F^H\bm H_{23}^H\!(\bm H_{23}\bm F\!\bm F^H\bm H_{23}^H\!+\! \bm I)^{-1}\bm H_{23}\bm F(\bm I\!+\!\bm B)^H-\Sigma_{k=1}^{N_T} \bm S_k^T\bm \mu_k^*\bm e_i^T=\bm 0,
\end{align}
which can be further arranged as \eqref{EQ_B_solution} with
\eqref{EQ_Delta1} and \eqref{EQ_Delta2}, where the matrix inversion
lemma is used. To get $\bm \mu_i$ we first apply the constraint
\eqref{EQ_constraint_lower_triangular} to $\bm B$ in
\eqref{EQ_B_solution}. Then we have
\begin{align}
\bm S_k(\Delta_1+\Sigma\bm S_i^T\bm \mu_i^*\bm e_i^T)(\bm I+ \Delta_2)\bm e_k=\bm 0,\,\,k=1\sim N_T,\notag
\end{align}
where $\Delta_1$ and $\Delta_2$ are defined in \eqref{EQ_Delta1} and \eqref{EQ_Delta2}, respectively.
After expanding we can get
\begin{align}
\bm S_k\Delta_1(\bm I+ \Delta_2)\bm e_k+\bm S_k\left(\underset{i}\Sigma\bm S_i^T\bm \mu_i^*\bm e_i^T\right)(\bm I+ \Delta_2)\bm e_k=\bm 0,\,\,k=1\sim N_T.\notag
\end{align}
Then combine the terms of $\bm e_k$ we can get
\begin{align}
\bm S_k\bm S_k^T\bm \mu_k^* + \bm S_k\left(\underset{i}\Sigma\bm
S_i^T\bm \mu_i^*\bm e_i^T\right)\Delta_2\bm e_k+\bm c_k=\bm
0,\,\,k=1\sim N_T,\notag
\end{align}
where the first term is due to the fact that $\bm e_i^T\bm e_k=\delta(i-k)$, $\delta$ is the Dirac delta function.
After combining the terms of $\bm \mu_i^*$, we have
\eqref{EQ_mu_equations}. Note that for each $k$, there are $k$
scalar equations and thus there are $(N_T+1)N_T/2$ equations in total.
Meanwhile, there are also $(N_T+1)N_T/2$ scalar unknown variables in
the set of variables $\{\bm \mu_1^*\cdots\bm\mu_{N_T}^*\}$. Thus we
can solve the multipliers from the linear equations at hand.

%In the following we give a simple example with $N_T=3$, in which
%$\bm S_1=[1,\,0,\,0],\,\bm S_2 =[1,\,0,\,0;0,\,1,\,0], \,\bm S_3=\bm I_3$.
%After expanding and combining the terms of $\bm \mu_i^*$ in \eqref{EQ_mu_equations}, we can get
%\begin{align}\label{EQ_matrix_equation}
%&\left(\!\!
%  \begin{array}{ccc}
%    1+\Delta_{2,11} & \Delta_{2,21}\bm S_1\bm S_2^T & \Delta_{2,31}\bm S_1\bm S_3^T \\
%    \Delta_{2,12}\bm S_2\bm S_1^T & (1+\Delta_{2,22})\bm I_2 & \Delta_{2,32}\bm S_2\bm S_3^T \\
%    \Delta_{2,13}\bm S_3\bm S_1^T & \Delta_{2,23}\bm S_3\bm S_2^T & (1+\Delta_{2,33})\bm I_3 \\
%  \end{array}
%\!\!\right)\!\!\!\!\left(\!\!
%         \begin{array}{c}
%           \!\!\bm \mu_1^* \!\!\\
%           \!\!\bm \mu_2^* \!\!\\
%           \!\!\bm \mu_3^* \!\!\\
%         \end{array}
%       \!\!\right)\!\!=\!\!\left(\!\!
%                 \begin{array}{c}
%                   \!\!-\bm c_1 \!\!\\
%                   \!\!-\bm c_2 \!\!\\
%                   \!\!-\bm c_3 \!\!\\
%                 \end{array}\!\!
%               \right).
%\end{align}
%It can easily be seen that both $\bm c_1,\, \bm c_2, \mbox{ and }\bm
%c_3,$ and $\bm \mu_1^*,\,\bm \mu_2^*, \mbox{ and }\bm \mu_3^*,$ have
%dimensions $1, 2,$ and $3$, respectively. Thus we can easily solve
%the multipliers through the above matrix equation. Note that the
%matrix inverse always exists due to the additional term $1$ in the
%diagonal of the matrix on the LHS of \eqref{EQ_matrix_equation}.

\section{Proof of Lemma \ref{Lemma_Taylor}}\label{Sec_Proof_lemma_Taylor}
It is known that the function log(det(.)) is concave in its domain
and the term inside the determinant of the numerator on the LHS of
\eqref{EQ_coexistence_constraint3} is a quadratic form of $\bm F$
and $\bm A$, which is convex. Thus the convexity of the composition
of log(det(.)) and the quadratic form is unknown. To proceed, we
first use the inequality \cite{Li_det_trace_ineq} $\mbox{det}(\bm I+\bm M)\geq 1+\mbox{tr}(\bm M)$,
i.e., the numerator on the LHS of \eqref{EQ_coexistence_constraint3}
can be lower bounded by
\begin{align}
&\mbox{log}\left( 1+\mbox{tr}(\bm H_{24}\bm F\bm F^H\bm
H_{24}^H+(\bm H_{14}+\bm H_{24}\bm A)\Sigma_{\bm x_P}(\bm H_{14}+\bm
H_{24}\bm
A))\right)\notag\\
\overset{(a)}=&\mbox{log}\left(1+\mbox{vec}(\bm H_{24}\bm
F)^H\mbox{vec}(\bm
H_{24}\bm F)+\mbox{vec}(\bm H_{14}+\bm H_{24}\bm A)^H(\bm I\otimes
\Sigma_{\bm x_P})\mbox{vec}(\bm H_{14}+\bm H_{24}\bm
A))\right)\notag\\
\overset{(b)}=&\mbox{log}\left(1+t_1+t_2\right),\label{EQ_log_t1_t2}
\end{align}
where in (a) we use the properties $\mbox{tr}(\bm M\bm
M^H)=\mbox{vec}(\bm M)^H\mbox{vec}(\bm M)$ and $\mbox{tr}(\bm M\bm
N\bm M^H)=\mbox{vec}(\bm M)^H(\bm I\otimes \bm N)\mbox{vec}(\bm M)$;
in (b) we introduce the following slack variables $t_1\triangleq\mbox{vec}(\bm
H_{24}\bm F)^H\mbox{vec}(\bm H_{24}\bm F)$ and $t_2\triangleq \mbox{vec}(\bm H_{14}+\bm
H_{24}\bm A)^H(\bm I\otimes\Sigma_{\bm x_P})\mbox{vec}(\bm H_{14}+\bm H_{24}\bm A)$. Then we
use the \textit{matrix-lifting semi-definite relaxation}
\cite{Ding_lifting} to relax the above such that $t_1\geq\mbox{vec}(\bm
H_{24}\bm F)^H\mbox{vec}(\bm H_{24}\bm F)$ and $t_2\geq \mbox{vec}(\bm H_{14}+\bm
H_{24}\bm A)^H(\bm I\otimes\Sigma_{\bm x_P})\mbox{vec}(\bm H_{14}+\bm H_{24}\bm A)$, to
accommodate the use of CVX. After this relaxation, we can further
rearrange them as the linear matrix inequalities as shown in
\eqref{EQ_t1_LMI}, \eqref{EQ_t2_LMI}, and \eqref{EQ_A3_LMI},
respectively by the Schur's Complement Theorem.

After replacing the numerator on the LHS of
\eqref{EQ_coexistence_constraint3} by
\eqref{EQ_t1_LMI}, \eqref{EQ_t2_LMI}, and \eqref{EQ_A3_LMI} and substituting
$\bm F_2\triangleq\bm F\bm F^H$ (which is further relaxed by $\bm
F_2\succeq\bm F\bm F^H$) into the denominator on the LHS of
\eqref{EQ_coexistence_constraint3}, we can observe that the LHS of
\eqref{EQ_coexistence_constraint3} is the difference of two concave
functions, and the convexity is still unknown. However, we may
resort to the skill in \cite{pslin_CR} to do the first order Taylor
expansion to the denominator of the LHS of
\eqref{EQ_coexistence_constraint3} and which will become a linear
function of $\bm F_2$. By definition, we know that the first order Taylor
expansion of the function $f(\bm X)$ with respect to $\bm X$ at
$\tilde{\bm X}$ is $ f(\bm X)\leq f(\bm X_0) +\mbox{tr}(\bm D^H(\bm X - \tilde{\bm X}))$,
where $\bm D\triangleq\frac{\partial f}{\partial \bm X}\Big|_{\bm X=\tilde{\bm X}}. $
Then with the fact $\partial \log\det(\bm X)/\partial \bm X=(\bm
X^{-1})^T$, after applying the above to the denominator on the LHS
of \eqref{EQ_coexistence_constraint3}, we can get
\begin{align}&\log\!|\!\bm I \!+\!\bm H_{24}\bm F_2\bm H_{24}^H\!|\!\leq\! \log(\!|\!\bm I\!+\!\bm H_{24}\!\tilde{\bm F}_2\!\bm \!H_{24}^H\!|\!)\!+\!\mbox{tr}\!\left(\!(\bm I\!+\!\bm H_{24}\!\tilde{\bm F}_2\!\bm H_{24}^H\!)^{-1}\!\bm H_{24}\!\bm F_2\!\bm H_{24}^H\!\right)\!-\!\mbox{tr}\left(\!(\bm I\!+\!\bm H_{24}\!\tilde{\bm F}_2\!\bm
H_{24}^H\!)^{-1}\!\bm H_{24}\!\tilde{\bm F}_2\!\bm H_{24}^H\!\right).\notag
\end{align}
After substituting the above into \eqref{EQ_coexistence_constraint3}
and some arrangements, we can get Lemma \ref{Lemma_Taylor}.

%\begin{align}\label{EQ_d14}
%\chi_{24}^2\left( d_c^2+c_1\right)=\chi_{24}^2\left( d_c^2+\frac{\lambda_{max}(\Sigma_{\bm x_P})}{\alpha P_T d_{24,max}^2}\right)\leq \lambda_{max}(\Sigma_{\bm x_P}),
%\end{align}
%then $\alpha$ increases with increasing tr$(\bm D_{14}^2)$, vice versa, to make the approximated coexistence constraint valid. Let us consider a more conservative coexistence constraint by removing the term $d_c^2$. With the definition of $\chi_{24}$, \eqref{EQ_d14} becomes $P_T\cdot d_{24,min}^2\geq 1,$ where $\alpha$ is eliminated by the numerator of $\alpha/(1-\alpha)$ on the RHS of (p).
%Therefore, we can write
%\begin{align}
%&\log|\bm I \!+\! (\bm I \!+\! \bm H_{24}\bm F\bm F^H\bm H_{24}^H)^{-1}(\bm
%H_{14}\!+\!\bm H_{24}\bm A)\Sigma_{\bm x_P}(\bm H_{14}\!+\!\bm H_{24}\bm
%A)^H|\!\varpropto\! \log\!\left(\!\!{\alpha}\cdot P_P^{-1}\cdot \chi_{24}^2 d_c^2(\mbox{tr}(\bm D_{14}^2))^{1-2\cdot \mathds{1}_{P_T\cdot d_{24,min}^2\geq 1.}}\!\right).\label{EQ_proportiona}
%\end{align}

%Then we can summarize that when SU is active, the upper bound of PU's rate is proportional to the logarithm of the product of the relay ratio, SU's transmit power, and the largest principal angle by Definition \ref{Def_principal_angle}.
\end{proof}
\bibliographystyle{IEEEtran}

\renewcommand{\baselinestretch}{1.5}
\bibliography{IEEEabrv,SecrecyPs2}

\begin{thebibliography}{10}
\providecommand{\url}[1]{#1}
\csname url@rmstyle\endcsname
\providecommand{\newblock}{\relax}
\providecommand{\bibinfo}[2]{#2}
\providecommand\BIBentrySTDinterwordspacing{\spaceskip=0pt\relax}
\providecommand\BIBentryALTinterwordstretchfactor{4}
\providecommand\BIBentryALTinterwordspacing{\spaceskip=\fontdimen2\font plus
\BIBentryALTinterwordstretchfactor\fontdimen3\font minus
  \fontdimen4\font\relax}
\providecommand\BIBforeignlanguage[2]{{%
\expandafter\ifx\csname l@#1\endcsname\relax
\typeout{** WARNING: IEEEtran.bst: No hyphenation pattern has been}%
\typeout{** loaded for the language `#1'. Using the pattern for}%
\typeout{** the default language instead.}%
\else
\language=\csname l@#1\endcsname
\fi
#2}}

\bibitem{Mitola_CR}
J.~Mitola, ``Cognitive radio: an integrated agent architecture for software
  defined radio,'' \emph{Ph.D. dissertation, KTH Royal Inst. Technology,
  Stockholm, Sweden}, 2000.

\bibitem{Haykin_CR}
S.~Haykin, ``Cognitive radio: brain-empowered wireless communications,''
  \emph{{IEEE} J. Select. Areas Commun.}, vol. 223, no.~2, pp. 201--220, Feb.
  2005.

\bibitem{80222}
``{IEEE} {P}802.22/{D}1.0 draft standard for wireless regional area networks
  part 22: Cognitive wireless {RAN} medium access control ({MAC}) and physical
  layer ({PHY}) specifications: Policies and procedures for operation in the
  {TV} bands,'' Apr. 2008.

\bibitem{devroye_CR}
P.~M. N.~Devroye and V.~Tarokh, ``Achievable rates in cognitive radio
  channels,'' \emph{{IEEE} Trans. Inform. Theory}, vol.~52, no.~5, pp.
  1813--1827, May 2006.

\bibitem{Jovicic_CR}
A.~Jovicic and P.~Viswanath, ``Cognitive radio: {a}n information-theoretic
  perspective,'' \emph{{IEEE} Trans. Inform. Theory}, vol.~55, no.~9, pp.
  3945--3958, Sep. 2009.

\bibitem{Wu_DPC_CR}
W.~Wu, S.~Vishwanath, and A.~Arapostathis, ``Capacity of a class of cognitive
  radio channels: {I}nterference channels with degraded message sets,''
  \emph{{IEEE} Trans. Inform. Theory}, vol.~53, no.~11, pp. 4391--4399, Nov.
  2007.

\bibitem{Rini_DPC_CR}
S.~Rini, D.~Tuninetti, and N.~Devroye, ``Inner and outer bounds for the
  {G}aussian cognitive interference channel and new capacity results,''
  \emph{{IEEE} Trans. Inform. Theory}, vol.~58, no.~2, pp. 820--848, Feb. 2012.

\bibitem{THP1}
M.~Tomlinson, ``New automatic equalizer employing modulo arithmetic,''
  \emph{Electr. Let.}, vol.~7, pp. 138--139, Mar. 1971.

\bibitem{THP2}
M.~Miyakawa and H.~Harashima, ``A method of code conversion for a digital
  communication channel with intersymbol interference,'' \emph{Trans. Inst.
  Elec. Comm. Eng. Japan}, vol. 52-A, pp. 272--273, Jun. 1969.

\bibitem{ShamaiMultibinning}
R.~Zamir, S.~Shamai, and U.~Erez, ``Nested linear/lattice codes for structured
  multiterminal binning,'' \emph{{IEEE} Trans. Inform. Theory}, vol.~48, no.~6,
  pp. 1250--1276, June 2002.

\bibitem{caire_channel_with_SI}
G.~Caire and S.~Shamai, ``On the capacity of some channels with channel state
  information,'' vol.~45, no.~6, pp. 2007--2019, Sept. 1999.

\bibitem{Shenouda_THP}
M.~B. Shenouda and T.~N. Davidson, ``A framework for designing {MIMO} systems
  with decision feedback equalization or {T}omlinson-{H}arashima precoding,''
  \emph{{IEEE} J. Select. Areas Commun.}, vol.~26, no.~2, pp. 401--411, Feb.
  2008.

\bibitem{Stankovic_THP}
V.~Stankovic and M.~Haardt, ``Generalized design of multi-user {MIMO} precoding
  matrices,'' \emph{{IEEE} Trans. Wireless Commun.}, vol.~7, no.~3, pp.
  953--961, Mar. 2008.

\bibitem{Yu_trellis_precoding}
W.~Yu, D.~P. Varodayan, and J.~M. Cioffi, ``Trellis and convolutional precoding
  for transmitter-based interference presubtraction,'' \emph{{IEEE} Trans.
  Commun.}, vol.~53, no.~7, pp. 1220--1230, July 2005.

\bibitem{Zhang_ZF_CR}
R.~Zhang and Y.-C. Liang, ``Exploiting multi-antennas for opportunistic
  spectrum sharing in cognitive radio networks,'' vol.~2, no.~1, pp. 88--102,
  Feb. 2008.

\bibitem{Luca_ZF_CR}
L.~Bixio, G.~Oliveri, M.~Ottonello, M.~Raffetto, and C.~S. Regazzoni,
  ``Cognitive radios with multiple antennas exploiting spatial opportunities,''
  \emph{{IEEE} Trans. Signal Processing}, vol.~58, no.~8, pp. 4453--4459, Aug.
  2010.

\bibitem{Gharavol_BF_CR}
E.~A. Gharavol, Y.~C. Liang, and K.~Mouthaan, ``Robust linear transceiver
  design in {MIMO} {A}d {H}oc cognitive radio networks with imperfect channel
  state information,'' \emph{{IEEE} Trans. Wireless Commun.}, vol.~10, no.~5,
  pp. 1448--1457, May 2011.

\bibitem{Hamdi_CR}
K.~Hamdi, K.~Zarifi, K.~B. Letaief, and A.~Ghrayeb, ``Beamforming in
  relay-assisted cognitive radio systems: {A} convex optimization approach,''
  in \emph{IEEE ICC 2011}, .

\bibitem{Zheng_CR}
G.~Zheng, H.~Song, K.~K. Wong, and B.~Ottersten, ``Cooperative cognitive
  networks: optimal, distributed and low-complexity algorithms,'' \emph{{IEEE}
  Trans. Signal Processing}, vol.~61, no.~11, pp. 2778 -- 2790, June 2013.

\bibitem{cvx}
M.~Grant and S.~Boyd, ``{CVX}: matlab software for disciplined convex
  programming, version 2.0 beta,'' \url{http://cvxr.com/cvx}, Sep. 2012.

\bibitem{80216}
``{IEEE} {S}td 802.16.-2012 {IEEE} standard for air interface for broadband
  wireless access systems,'' 2012.

\bibitem{Costa_DPC}
M.~H.~M. Costa, ``Writing on dirty paper,'' \emph{{IEEE} Trans. Inform.
  Theory}, vol.~29, pp. 439--441, May 1983.

\bibitem{pslin_CR}
P.-H. Lin, S.-C. Lin, C.-P. Lee, and H.-J. Su, ``Cognitive radio with partial
  channel state information at the transmitter,'' \emph{{IEEE} Trans. Wireless
  Commun.}, vol.~9, no.~11, pp. 3402--3413, Nov. 2010.

\bibitem{GPC}
S.~I. Gelfand and M.~S. Pinsker, ``Coding for channel with random parameters,''
  \emph{Problems of Control and Inf. Theory}, vol.~9, no.~1, pp. 19--31, 1980.

\bibitem{PSLIN_CRRX2}
P.-H. Lin, S.-C. Lin, H.-J. Su, and Y.-W.~P. Hong, ``Improved transmission
  strategies for cognitive radio under the coexistence constraint,''
  \emph{{IEEE} Trans. Wireless Commun.}, vol.~11, no.~11, pp. 4058--4073, Nov
  2012.

\bibitem{SC_JSAC}
S.-C. Lin and H.-J. Su, ``Practical vector dirty paper coding for {MIMO}
  {G}aussian broadcast channels,'' \emph{{IEEE} J. Select. Areas Commun.},
  vol.~25, no.~7, pp. 1345--1357, Sep. 2007.

\bibitem{Erez_DPC_code_design}
U.~Erez and S.~ten. Brink, ``A close to capacity dirty paper coding scheme,''
  \emph{{IEEE} Trans. Inform. Theory}, vol.~51, no.~10, pp. 3417--3432, Oct.
  2005.

\bibitem{Golub_principal_angle}
A.~Bjorck and G.~H. Golub, ``Numerical methods for computing angles between
  linear subspaces,'' \emph{Math. Comput.}, vol.~27, no. 123, pp. 579--594,
  July 1973.

\bibitem{Conway_chordal_distance}
J.~H. Conway, R.~H. Hardin, and N.~J.~A. Sloane, ``Packing lines, planes, etc.:
  packings in {G}rassmannian space,'' \emph{Experimental Mathematics}, vol.~5,
  pp. 139--159, 1996.

\bibitem{Boyd_subgradient}
S.~Boyd, L.~Xiao, and A.~Mutapcic, ``Subgradient methods,'' \emph{Notes for EE
  392o, Stanford University}, Oct. 1 2003.

\bibitem{Fischer_precoding}
R.~F.~H. Fischer, \emph{Precoding and signal shaping for digital
  transmission}.\hskip 1em plus 0.5em minus 0.4em\relax New York: Wiley, 2002.

\bibitem{Li_MSE_SINR}
L.~Li, Y.-D. Yao, and H.~Li, ``Transmit diversity and linear and
  decision-feedback equalizations for frequency-selective fading channels,''
  \emph{{IEEE} Trans. Veh. Technol.}, vol.~52, no.~5, pp. 1217--1231, Sept.
  2003.

\bibitem{Cover_book}
T.~M. Cover and J.~A. Thomas, \emph{Elements of information theory},
  2nd~ed.\hskip 1em plus 0.5em minus 0.4em\relax New York: Wiley, 2006.

\bibitem{Horn_matrix_analysis}
R.~A. Horn and C.~R. Johnson, \emph{Matrix analysis}.\hskip 1em plus 0.5em
  minus 0.4em\relax Cambridger University Press, 1985.

\bibitem{Kermoal_MIMO_channel_model}
J.~P. Kermoal, L.~Schumacher, K.~I. Pedersen, P.~E. Mogensen, and
  F.~Frederiksen, ``A stochastic {MIMO} radio channel model with experimental
  validation,'' \emph{{IEEE} J. Select. Areas Commun.}, vol.~20, no.~6, pp.
  1211--1226, Aug. 2002.

\bibitem{Hamdi_ZF_CR}
K.~Hamdi, W.~Zhang, and K.~B. Letaief, ``Opportunistic spectrum sharing in
  cognitive {MIMO} wireless networks,'' \emph{{IEEE} Trans. Wireless Commun.},
  vol.~8, no.~8, pp. 4098--4109, Aug. 2009.

\bibitem{Heath_MIMO_CR}
K.~Lee, C.~B. Chae, J.~R.~W.~Heath, and J.~Kang, ``{MIMO} transceiver design
  for spatial sensing in cognitive radio networks,'' \emph{{IEEE} Trans.
  Wireless Commun.}, vol.~10, no.~11, Nov. 2011.

\bibitem{Zhang_matrix_theory}
F.~Zhang, \emph{Matrix theory: basic and techniques}.\hskip 1em plus 0.5em
  minus 0.4em\relax Springer-Verlag, New York, 1999.

\bibitem{Kumar_matrix_inequality}
J.~K. Merikoski and R.~Kumar, ``Inequalities for spreads of matrix sums and
  products,'' \emph{Applied Mathematics E-Notes}, pp. 150--159, 2004.

\bibitem{Li_det_trace_ineq}
Q.~Li and W.~K. Ma, ``Optimal and robust transmit designs for {MISO} channel
  secrecy by semidefinite programming,'' \emph{{IEEE} Trans. Signal
  Processing}, vol.~59, no.~8, pp. 3799--3812, Aug. 2011.

\bibitem{Ding_lifting}
Y.~Ding and H.~Wolkowicz, ``A matrix-lifting semidefinite relaxation for the
  quadratic assignment problem,'' \emph{Department of Combinatorics and
  optimization, University of Waterloo, Tech Rep. CORR 06-22}, 2006.

\end{thebibliography}
\newpage

\end{document}